\documentclass[a4paper,UKenglish, cleveref, autoref, thm-restate]{lipics-v2021}

\usepackage{complexity}
\usepackage{xspace}
\usepackage{graphics}
\usepackage[ruled]{algorithm}
\usepackage{graphicx}
\usepackage{xargs}
\usepackage{mathtools}
\usepackage{framed}
\hypersetup{colorlinks=true, citecolor=blue}

\title{Dynamic Planar Embedding is in DynFO} 

\author{Samir Datta}{Chennai Mathematical Institute \& UMI ReLaX, Chennai, India}{sdatta@cmi.ac.in}{https://orcid.org/0000-0003-2196-2308}{}

\author{Asif Khan}{Chennai Mathematical Institute, Chennai, India}{asifkhan@cmi.ac.in}{}{}

\author{Anish Mukherjee}{University of Warwick, Coventry, United Kingdom}{anish.mukherjee@warwick.ac.uk}{https://orcid.org/0000-0002-5857-9778}{}

\authorrunning{S.\, Datta, A.\, Khan, and A.\, Mukherjee} 

\Copyright{Samir Datta, Asif Khan, and Anish Mukherjee} 



\category{} 

\relatedversion{} 

\acknowledgements{Thanks to Nils Vortmeier and Thomas Zeume for illuminating discussions. SD,AK: Partially funded by a grant from Infosys foundation AM: Research supported in part by the Centre for Discrete Mathematics and its Applications (DIMAP), by EPSRC award EP/V01305X/1.}

\nolinenumbers 

\hideLIPIcs  

\EventEditors{John Q. Open and Joan R. Access}
\EventNoEds{2}
\EventLongTitle{42nd Conference on Very Important Topics (CVIT 2016)}
\EventShortTitle{CVIT 2016}
\EventAcronym{CVIT}
\EventYear{2016}
\EventDate{December 24--27, 2016}
\EventLocation{Little Whinging, United Kingdom}
\EventLogo{}
\SeriesVolume{42}
\ArticleNo{23}

\newcommand{\es}[3][1=\sigma]{[\smash{{#2}\xrightarrow{#1}{#3}}]}

\newcommand{\maxcon}{actually}
\newcommand{\sqline}{\leftrightsquigarrow}

\newcommand{\db}{\mathcal{I}}
\newcommand{\schema}{\sigma}
\newcommand{\state}{\mathcal{S}}
\newcommand{\aux}{\mathcal{A}}
\newcommand{\prog}{\mathcal{P}}
\newcommand{\query}{Q}
\newcommand{\ans}{Ans}
\newcommand{\domain}{D}

\begin{document}
\maketitle

\begin{abstract}
Planar Embedding is a drawing of a graph on the plane such that the edges do not intersect each other except at the vertices. We know that testing the planarity of a graph and computing its embedding (if it exists), can efficiently be computed, both sequentially~\cite{HopcroftTarjan} and in parallel~\cite{RR94}, when the entire graph is presented as input.  

In the dynamic setting, the input graph changes one edge at a time through insertion and deletions and planarity testing/embedding has to be updated after every change. 
By storing auxilliary information we can improve the complexity of dynamic planarity testing/embedding over the obvious recomputation from scratch. In the sequential dynamic setting, there has been a series of works~\cite{EGIS,IPR,HIKLR,HR1}, culminating in the breakthrough result of $\polylog(n)$ sequential time (amortized) planarity testing algorithm of Holm and Rotenberg~\cite{HR2}.  

In this paper we study planar embedding through the lens of $\DynFO$, a parallel dynamic complexity class introduced by Patnaik et al.~\cite{PI} (also~\cite{DST95}).
We show that it is possible to dynamically maintain whether an edge can be inserted to a planar graph without causing non-planarity in $\DynFO$. We extend this to show how to maintain an embedding of a planar graph under both edge insertions and deletions, while rejecting edge insertions that violate planarity. 

Our main idea is to maintain embeddings of only the triconnected components and a special two-colouring of separating pairs that enables us to side-step cascading flips when embedding of a biconnected planar graph changes, a major issue for sequential dynamic algorithms~\cite{HR1, HR2}.
\end{abstract}
\section{Introduction}
Planar graphs are graphs for which there exists an embedding of vertices on the plane such that the edges can be drawn without intersecting with each other, except at their endpoints. The notion of planar graphs is fundamental to graph theory as underlined by the Kuratowski theorem~\cite{kuratowski30}.
The planarity testing problem is to determine if the encoded graph is planar and the planar embedding problem is to construct such an embedding. These are equally fundamental questions to computer science and their importance has been recognized from the early 1970s in the linear time algorithm by Hopcroft and Tarjan~\cite{HopcroftTarjan}. Since then there have been a plethora of algorithmic solutions presented for the planarity testing and embedding problems such as \cite{LEC, BL, ET} that culminated in an alternative linear-time algorithm~\cite{ET}, a work efficient parallel algorithm running in $O(\log{n})$ time \cite{RR94}, a deterministic logspace algorithm \cite{AM, DP}, and many more.

All of the above algorithms are static i.e., the input is presented at once and we need to answer the planarity testing query and produce an embedding only once. However, in many real-life scenarios, the input is itself dynamic and evolves by insertion and deletion of edges. The same query can be asked at any instance and even the embedding may be required. Rather than recomputing the result from scratch after every update to the input in many scenarios, it is advantageous to preserve some auxiliary data such that each testing or embedding query can be answered much faster than recomputation from scratch. The notion of ``fast'' can be quantified via the sequential time required to handle the updates and queries which can be achieved in polylogarithmic time as in the recent breakthrough works of \cite{HR1, HR2}. These in turn built upon the previous work that dealt with only a partially dynamic model of computation -- insertion only \cite{Poutre, BT96, Westbrook} or, deletion only \cite{IPR} or the fully dynamic model (that supports both insertions and deletions) but with polynomial time updates~\cite{EGIS}.

Our metric for evaluating updates is somewhat different and determined according to the Dynamic Complexity framework of Immerman and Patnaik \cite{PI}, and is closely related to the setting of Dong, Su, and Topor \cite{DST95}. In it, a dynamic problem is characterised by the smallest complexity class in which it is possible to place the updates to the auxiliary database and still be able to answer the queries (notice that if the number of possible queries is polynomial we can just maintain the answers of all queries in the auxiliary database).
 
Notable amongst these has been the first-order logic formulas or equivalently, the descriptive complexity class $\FO$. Thus we obtain the class $\DynFO$ of dynamic problems for which the updates to the auxiliary data structure are in $\FO$ given the input structure and stored auxiliary data structures. The motivation to use first-order logic as the update method has connections to other areas as well e.g., it implies that such queries are highly parallelisable, i.e., can be updated by polynomial-size circuits in constant-time due to the correspondence between $\FO$ and uniform $\ACz$ circuits~\cite{BIS}. From the perspective of database theory, such a program can be translated into equivalent SQL queries.

A particular recent success story of dynamic complexity has been that directed graph reachability (which is provably not in $\FO$) can be maintained in $\DynFO$ \cite{DKMSZ} resolving a conjecture from Immerman and Patnaik \cite{PI}, open since the inception of the field. Since then, progress has been made in terms of the size of batch updates (i.e., multiple simultaneous insertions and deletions) that can be handled for reachability, distance, and maximum matching \cite{DMVZ, M19}. Later, improved bounds have been achieved for these problems in various special graph classes, including in planar graphs \cite{DKMTVZ, DGJMST22}. Problems in planar graphs have been studied in the area of dynamic complexity starting much earlier e.g., before the reachability conjecture was resolved, it was shown in~\cite{DHK14} that reachability in embedded planar graphs is in $\DynFO$. Also, in \cite{Mehta14} it was shown that 3-connected planar graph isomorphism too is in $\DynFO$ with some precomputation. However, despite these works the dynamic complexity of the planarity testing problem itself is not yet resolved, let alone maintaining a planar embedding efficiently. 

\subparagraph*{Our contribution}
In this paper, we build on past work in dynamic complexity to show that a planar embedding 
can be maintained efficiently, where we test for planarity at every step. 
Here, by planar embedding we mean a cyclic order on the neighbours of every
vertex in some drawing of the graph on the plane (also known as combinatorial
embedding~\cite{Djidjev}).
\begin{restatable}{theorem}{main}
Given a dynamic graph undergoing insertion and deletion of edges we can 
maintain a planar embedding of the graph in $\DynFO$ (while never allowing 
insertion of edges that cause the graph to become non-planar).
\end{restatable}

\subparagraph*{Organization.}
We start with preliminaries concerning graph theory and dynamic complexity in \cref{sec:prelims}. 
We present a technical overview of our work in \cref{sec:technical}. 
In \cref{sec:graphTh} we develop the graph theoretic machinery we need for our algorithm. 
In Section~\ref{sec:query} we formalize the query model and the auxiliary data stored. 
We  describe the implementation of the connectivity data structures which we detail in Section~\ref{sec:app:triconn}. 
Next, we give an overview of the dynamic planar embedding algorithm in Section~\ref{sec:algoOverview}. 
In Section~\ref{sec:app:graphThApp} we provide the details left out on graph theoretic machinery. 
In Section~\ref{sec:planarFO} we introduce the primitives required in the subsequent sections for maintaining planarity and argue that they can be implemented in $\FO$. In Section~\ref{sec:triChange} we describe the maintenance of the planar embedding of triconnected components. This last invokes and is used to maintain the two-colouring of separating pairs which is described in Section~\ref{sec:coloring}. We show how to maintain a planar embedding of biconnected components and extend it to a planar embedding of the entire graph in \cref{sec:embed}. 

\section{Preliminaries}\label{sec:prelims}
We start with some notations followed by graph theoretic preliminaries related to connectivity and planarity -- see \cite[Chapters 3, 4]{Diestel} for a thorough introduction.
Then we reproduce some essentials of Dynamic Complexity from \cite{DKMTVZ,DMVZ}.

Given a graph $G = (V,E)$, we write $V(G)$ and $E(G)$ to denote the sets of vertices and edges of $G$, respectively. For a set of edges $S \subseteq E(G)$ we denote by $G - S$, the graph with the edges in $S$ deleted.  Similarly for $S \subseteq V(G) \times V(G)$ we denote by $G + S$ the graph to which new edges in $S$ have been added. For a set of vertices $T \subseteq V(G)$, by $G - T$ we refer to the induced graph $G[V(G) \setminus T]$. An undirected path between $u$ and $v$ is denoted by $u\sqline v$. 

\subsection{Biconnected and Triconnected Decomposition}
We assume familiarity with common connectivity related terminology 
including $2$-vertex connectivity, $3$-vertex connectivity and the
related separating sets viz. cut vertices, separating pairs and the
notion of virtual edges in the triconnected decomposition. 

\paragraph*{Biconnectivity and Biconnected Decomposition} A vertex of a connected graph is called a \emph{cut-vertex} if deleting it from the graph disconnects the graph. A graph is called \emph{biconnected} or $2$-connected graph if there is no cut vertex in it. A connected graph can be decomposed into its maximal \emph{biconnected components} such that two vertices are in one biconnected component if no cut vertex deletion can disconnect them. The decomposition can be expressed as tree which has nodes corresponding to biconnected components and the cut vertices. There is an edge between a cut vertex node and biconnected component node iff the cut vertex belongs to the biconnected component. The biconnected component nodes are termed B (for \emph{block}) and the cut vertex nodes are termed C (for cut vertex). The decomposition tree is called a BC-tree.

\paragraph*{Triconnectivity and Triconnected Decomposition} In a biconnected graph a pair of vertices is called a \emph{separating pair} if their deletion from the graph disconnects the graph. A graph is called \emph{$3$-vertex-connected} if there are no separating pairs in it. A separating pair is called $3$-connected if there is no separating pair that disconnects them. 3-connected separating pairs define a unique decomposition of any biconnected graph into triconnected components. Two vertices of the graph are in a triconnected component if there doesn't exist a 3-connected separating pair whose deletion disconnects them. In each triconnected component there are virtual edges corresponding to the 3-connected separating pairs that belong to the triconnected component apart from the actual graph edges. The decomposition can be expressed as a tree with four types of nodes: R-nodes that correspond to triconnected components that are 3-connected (rigid nodes), S-nodes that correspond to triconnected components that are cycles (serial nodes) and P-nodes that correspond to 3-connected separating pair nodes (parallel nodes). There is an edge between a P-node and an R-node (or S-node) if the vertices corresponding to the P-node belong to triconnected component corresponding to the R-node (or S-node). Note that, triconnected components are not the same as $3$-connected components, e.g, a cycle is a triconnected component but not a $3$-connected component. For more details see~\cite{BT96,HIKLR,DLNTW}. We also have Q-nodes that correspond to single edges that are bridges. They are easy to deal with and henceforth we will eschew any mention of them. 

Note the use of $3$-connected separating pairs. This is required because if we use any pair of vertices that are separating pair for the decomposition, it may be that this decomposition is not unique. As an example for cycles of length $\geq 4$ every chord is a separating pair and moreover ``interlacing'' chords will not allow a consistent way to form a $3$-connected-separating pair tree analogous to the block-cut-vertex tree. Thus conventional wisdom has it \cite{HopcroftTarjan, DLNTW} that we ignore interlacing separating pairs (i.e., separating pairs such that deleting one of them causes the other to become disconnected) and only consider the non-interlacing separating pairs.

The following are two data structures that help in representing tree
decompositions associated with biconnectivity and $3$-connectivity 
respectively.

\begin{enumerate}
	\item \textbf{BC-tree} or block-cut tree of a connected component of the graph, say $H$, denoted by $T_2(H)$. The nodes of the tree are the biconnected components (block nodes) and the cut vertices (cut nodes) of $H$ and the edges are only between cut and block nodes. Block nodes are denoted by $B$ and the cut nodes are denoted by $C$. 
	\item \textbf{SPQR-tree} or the triconnected decomposition tree of a biconnected component of the graph say $B$, is denoted by $T_3(B)$. The nodes in the SPQR-tree are of one of four types: $S$ denotes a cycle component (serial node), $P$ denotes a $3$-connected separating pair (parallel node), $Q$ denotes that there is just a single edge in $B$, and $R$ denotes the $3$-connected components or the so-called rigid nodes. There is an edge between an R-node, say $R_i$ and a P-node, say $P_j$ if $V(P_j)\subset V(R_i)$, and similarly, edges between S and P-nodes are defined. 
\end{enumerate}

We will conflate a node in one of the two trees with the corresponding subgraph. For example, an R-node interchangeably refers to the tree node as well the associated rigid subgraph.

\subsection{Planar Embedding}
A planar embedding of a graph $G=(V, E)$ is a mapping of vertices and edges in the plane $\mathbb{R}^2$ such that the vertices are mapped to distinct points in the plane and every edge is mapped to an arc between the points corresponding to the two vertices incident on it such that no two arcs have any point in common except at their endpoints. This embedding is called a \emph{topological embedding}. Corresponding to a given topological embedding, the faces of the graph are the open regions in $\mathbb{R}^2\setminus G$ (plane with points corresponding to the vertices and edges removed), call the set of faces as $F$. For a face $f \in F$, the set of all the vertices that lie on the boundary of $f$, is denoted by $V(f)$. The unbounded face is called the outer face. An embedding on the surface of a sphere is similarly defined. On the sphere, every face is bounded. Two topological embeddings are equivalent if, for every vertex, the cyclic order of its neighbours around the vertex is the same in both embeddings. So, the cyclic order (or rotation scheme) around each vertex defines an equivalence on the topological embeddings. The vertex rotation scheme around each vertex encodes the embedding equivalence class (combinatorial embedding). We now recall two important results. The first result says that a $3$-connected planar graph has unique planar embedding on the sphere (up to reflection).

\begin{theorem}[Whitney~\cite{Whitney}]\label{lem:whit}
	Any two planar embeddings of a $3$-connected graph are equivalent.
\end{theorem}

The second one is a criterion for planarity of biconnected graphs.
\begin{lemma}[Mac Lane~\cite{MacLane}]\label{lem:macl}
A biconnected graph is planar if and only if its triconnected components are planar.
\end{lemma}

\subsection{Dynamic Complexity} 
The goal of a dynamic program is to answer a given query on an \emph{input structure} subject to changes that insert or delete tuples. The program may use an auxiliary data structure represented by an \emph{auxiliary structure} over the same domain. Initially, both input and auxiliary structure are empty; and the domain is fixed during each run of the program.

For a (relational) structure $\db$ over domain $\domain$ and schema $\schema$, a change $\Delta \db$ consists of sets $R^{+}$ and $R^{-}$ of tuples for each relation symbol $R \in \schema$. The result $\db + \Delta \db$ of an application of the change $\Delta \db$  to $\db$ is the input structure where $R^{\db}$ is changed to $(R^{\db} \cup R^{+}) \setminus R^{-}$. The \emph{size} of $\Delta \db$ is the total number of tuples in relations $R^{+}$ and $R^{-}$ and the set of \emph{affected elements} is the (active) domain of tuples in $\Delta \db$.

\subparagraph*{Dynamic Programs and Maintenance of Queries.} A dynamic program 
consists of a set of update rules that specify how auxiliary relations are updated after changing the input structure. An \emph{update rule} for updating 
an $\ell$-ary auxiliary relation $T$ after a change is a first-order formula $\varphi$ over schema $\tau \cup \tau_{\text{aux}}$ with $\ell$ free variables, where $\tau_{\text{aux}}$ is the schema of the auxiliary structure. After a change $\Delta \db$, the new version of $T$ is 
$T := \{ \bar{a} \mid (\db + \Delta \db, \aux) \models \varphi(\bar{a})\}$ 
where $\db$ is the old input structure and $\aux$ is the current auxiliary 
structure. Note that a dynamic program can choose to have access to the old input structure by storing it in its auxiliary relations. 

For a state $\state = (\db, \aux)$ of the dynamic program $\prog$ with input structure $\db$ and auxiliary structure $\aux$, we denote by $\prog_\alpha(\state)$, the state of the program after applying a change sequence $\alpha$ and updating the auxiliary relations accordingly.
The dynamic program \emph{maintains} a $q$-ary query $\query$ under changes that affect $k$ elements (under changes of size $k$, respectively)  if it has a $q$-ary auxiliary relation $\ans$ that at each point stores the result of $\query$ applied to the current input structure. More precisely, for each non-empty sequence $\alpha$ of changes that affect $k$ elements (changes of size $k$, respectively), the relation $\ans$ in $\prog_\alpha(\state_\emptyset)$ and $\query(\alpha(\db_\emptyset))$ coincide, where $\db_\emptyset$ is an empty input structure, $\state_\emptyset$ is the auxiliary structure with empty auxiliary relations over the domain of $\db_\emptyset$, and $\alpha(\db_\emptyset)$ is the input structure after applying $\alpha$.
If a dynamic program maintains a query, we say that the query is in $\DynFO$. 

\section{Technical Overview}\label{sec:technical}
It is well known from Whitney's theorem (\cref{lem:whit}) that $3$-connected planar graphs are rigid i.e., they (essentially) have a unique embedding. Thus, for example, under the promise that the graph remains $3$-connected and planar it is easy to maintain an embedding in $\DynFO$ (see for example \cite{Mehta14}). An edge insertion occurs within a face and there are only local changes to the embedding -- restricted to a face. Deletions are exactly the reverse.

On the other extreme are trees, which are minimally connected. These are easy to maintain as well because any vertex rotation scheme is realisable. However, biconnected components are not rigid and yet not every rotation scheme for a vertex is valid (see \cref{subfig:two3b} for an illustration). The real challenge is in maintaining embeddings of biconnected components. 

This has been dealt with in literature by decomposing biconnected graphs into $3$-connected components (which are rigid components in the context of planar graphs). The $3$-connected components are organized into trees\footnote{The tree decomposition of a biconnected graph into $3$-connected pieces is a usual tree decomposition (\cite[Chapter 12.3]{Diestel})} popularly called SPQR-trees~\cite{BT96}. The approach is to use the rigidity of the $3$-connected planar components and the flexibility of trees to maintain a planar embedding of biconnected graphs. In order to maintain a planar embedding of connected graphs we need a further tree decomposition into biconnected components that yields the so-called block-cut trees or BC-trees (\cite[Lemma 3.1.4]{Diestel}, \cite{Harary}). Notice that the tree decomposition into SPQR-trees and BC-trees is Logspace hard~\cite{DLNTW} and hence not in $\FO$. Thus, in the parallel dynamic setting, we emulate previous sequential dynamic algorithms in \emph{maintaining} (rather than computing from scratch) SPQR-trees and BC-trees in our algorithm.

\subparagraph*{Issues with biconnected embedding.}
The basic problem with maintaining biconnected planar components 
is their lack of rigidity (with reaching complete flexibility). Thus insertion of an edge into a biconnected component might necessitate changing the embedding through operations called \emph{flips} and \emph{slides} in literature \cite{HR1,HR2} (see Figure~\ref{fig:flipsSlides}). We might need lots of flips and slides for a single edge insertion -- causing an exponentially large search space. We now proceed to describe these changes in more detail.

In the simplest form consider a biconnected graph and a separating pair contained within, separating the component into two $3$-connected components. We can reflect one of the $3$-connected components, that is, the vertex rotation for each vertex in the $3$-connected component is reversed. More intuitively, mirror a piece across one of its separating pairs. In more complicated cases there may be cascading flips i.e., reflections across multiple separating pairs in the biconnected component. We also need to deal with \emph{slides}, that is changes in ordering of the biconnected components at a separating pair. A single edge insertion might need multiple flips and slides.

This induces the definition of flip-distance i.e., the minimum number of flips and slides to change one embedding of the graph to another. Intuitively, the flip-distance lower bounds the sequential time needed to transition from one embedding to another.

Thus a crucial part of previous algorithms \cite{HR1} deals with maintaining an embedding of small flip distance with every possible embedding that can arise after a single change. In \cite{HR2} this algorithm is converted to a fully dynamic algorithm that handles updates in $O(\log^3{n})$ time using a sophisticated amortization  over the number of flips required to transition to an appropriate embedding dominates the running time of their algorithm. Notice that \cite{HR1} handles changes in $O(\log^3{n})$ worst case time but only in the incremental setting.

\subparagraph*{Our approach for dynamic planarity testing.}
We now switch to motivating our approach in the parallel dynamic setting that is fully dynamic and does not use amortization. There are fundamentally two issues to be resolved while inserting an edge -- one is whether the resulting graph is planar. The other is, how to update the embedding, possibly by performing multiple flips, when the graph remains planar. Let us first focus on biconnected graphs and the corresponding tree decomposition SPQR-tree introduced in Section~\ref{sec:prelims}.

To check for planarity on insertion of an edge, we introduce the notion of $P_i,P_j-$\emph{coherent paths} (\cref{def:coherent}). A path between two P-nodes $P_i,P_j$ in an SPQR-tree is said to be coherent if for every R-node $R_k$ on the path, all the (at most four) vertices of the two adjacent P-nodes are all on one face of $R_k$. This yields a combinatorial characterisation of coherent paths, given an embedding of each rigid component. The embedding of a rigid component is dealt with separately, later on. The significance of coherent paths stems from a crucial lemma (\cref{lem:ptest}). This shows that an edge $(a,b)$ is insertable in the graph preserving planarity if and only if the ``projection'' of any simple $a,b$ path in the graph\footnote{which satisfies a minimality condition -- it does not pass through both vertices of a separating pair} onto the corresponding SPQR-tree roughly corresponds to a coherent path. This yields a criterion for testing planarity after an edge insertion which can be implemented in $\FO$.

\subparagraph*{Insertions in biconnected components.}
Having filtered out non-planarity causing edges we turn to the question of how to construct the new planar embedding of the biconnected components after an edge insertion. The answer will lead us to investigate how to embed a rigid component when it is synthesized from a biconnected component.

It is in this context that we introduce the notion of two-colouring of separating pairs. This is a partial sketch of the new $3$-connected component formed after an edge insertion. More concretely, the separating pairs along the path in the SPQR-tree are no longer separating pairs after the edge insertion and the common face (as ensured by the crucial lemma alluded to above) on which all the endpoints of the previous separating pairs lie splits into two faces. Since the embedding of a $3$-connected planar graph is unique, after the edge insertion the two new faces formed are also unique, i.e., do not depend on the embedding. Thus the endpoints of each previous separating pair can be two coloured depending on which of the two faces a separating pair belongs to. We prove in the two-colouring lemma (\cref{lem:p2Col}) that no separating pair has both vertices coloured with the same colour.

Notice that when an edge is inserted such that its endpoints lie on a coherent path, all the rigid components on the path coalesce into one large rigid component (see \cref{subfig:two3a}). Two-colouring allows us to deal with flips by telling us the correct orientation of the coalescing rigid components on edge insertion. This, in turn, allows us to obtain the face-vertex rotation scheme of the modified component. In addition, it helps us to maintain the vertex rotation scheme in some corner cases (when two or more separating pairs share a vertex).

\subparagraph*{Face-vertex rotation scheme.} 
The sceptical reader might question the necessity of maintaining the face-vertex rotation scheme for a $3$-connected component. This is necessary for two reasons -- first, to apply the planarity test we need to determine the existence of a common face containing a $4$-tuple (or $3$-tuple) of vertices. The presence of a face-vertex rotation scheme directly shows that this part is in $\FO$. Second and more crucially, we need it to check if a particular triconnected component needs to be reflected after cascaded flips. Maintaining the vertex rotation scheme for biconnected components is now simple -- we just need to collate the vertex rotation schemes for individual rigid components into one for the entire graph.

\subparagraph*{Handling deletions.}
On deleting edges while it's not necessary to perform additional flips, the rest of the updates is roughly the reverse of insertion. On deleting an edge from a rigid component, we infer two-colourings from the embedding of erstwhile rigid components that decompose into pieces. Further, we have to update the coherent paths since possibly more edges are insertable preserving planarity. Notice that when an edge is deleted from a biconnected component this might lead to many simultaneous virtual edge deletions that might in turn cause triconnected components to decompose. Many ($O(n)$) invocations of the above triconnected edge deletion will be needed, but they can be handled in constant parallel time because they independent of each other as far as the updates required are concerned (see \cref{fig:ot}).

\subparagraph*{Extension to the entire embedding.}
BC-trees for connected components have blocks and cut vertices as their nodes. We can maintain an embedding for the graph corresponding to a block or B-node as above. Since a non-cut vertex belongs to precisely one block, we can inherit the rotation scheme for such vertices from that of the blocks. For cut vertices, we need to splice together the vertex rotation scheme from each block that the cut vertex is incident on as long as the order respects the ordering provided by individual blocks. 

\subparagraph*{Low level details of the information maintained.}
We maintain BC-tree for each connected component and SPQR-trees for each biconnected component thereof. In each of these trees we maintain betweenness information, i.e., for any three nodes $X_1, X_2$ and $X_3$ whether $X_2$ occurs on the tree path between $X_1$ and $X_3$. We also maintain a two-colouring of separating pairs for each $P_i, P_j$-coherent path in every SPQR-tree. For each rigid component $R_i$ and each cycle component $S_j$ we maintain their extended planar embedding. Specifically, we maintain the vertex rotation scheme in the following form. For every vertex $v$, we maintain triplet(s) $(v_i,v_j,v_k)$ of neighbours of $v$ that occur in the clockwise order though not necessarily consecutively. This enables us to insert and delete an arbitrarily large number of neighbours in $\FO$ making it crucial for the planar embedding procedure. This would not be possible if we were to handle individual insertions and deletions separately. See Figure~\ref{fig:deleteMulti}
for an example. We use a similar representation for the face-vertex rotation scheme. 

For biconnected components, we maintain only a planar embedding (not the extended version) since the face-vertex rotation scheme is not necessary.

\subparagraph*{Comparison with existing literature.}
The main idea behind recent algorithms for planar embedding in the sequential dynamic setting has been optimizing the number of flips necessitated by the insertion of an edge. This uses either a purely incremental algorithm or alternatively, a fully dynamic but amortized algorithm. Since our model of computation is fully dynamic and does not allow for amortization, each change must be handled (i.e., finding out the correct cascading flips) in worst case $O(1)$-time on $\CRCW$-PRAM. We note that filtering out edges that violate planarity in dynamic sequential $t(n)$ time (a \emph{test-and-reject} model) implies an amortized planarity testing algorithm with $O(t(n))$ time  (i.e., a \emph{promise-free} model). In contrast, although we have a test-and-reject model we are unable to relax the model to promise-free because of lack of amortization.

There are weaker promise models such as the one adopted in \cite{DMSVZ} where for maintaining a bounded tree-width decomposition it is assumed that the graph has tree-width at most $k$ without validating the promise at every step. In contrast our algorithm can verify the promise that no non-planarity-causing edge is added.

In terms of query model support, most previous algorithms \cite{HR1, HR2, IPR} only maintain the vertex rotation scheme in terms of clockwise next neighbour, in fact, \cite{HR1, HR2} need $O(\log n)$ time to figure out the next neighbour. In contrast, we maintain more information in terms of arbitrary triplets of neighbours in (not necessarily consecutive) clockwise order. This allows us to sidestep following arbitrarily many pointers, which is not in $\FO$. Finally, in terms of parallel time our algorithm (since it uses $O(1)$ time per query/update on $\CRCW$-PRAMs) is optimal in our chosen model. In contrast, the algorithm of \cite{HR2} comes close but fails to achieve the lower bound (of $\Omega(\log{n})$) in the sequential model of dynamic algorithms.

\section{Graph Theoretic Machinery}\label{sec:graphTh}
First, we present some graph theoretic results which will be crucial for our maintenance algorithm. We begin with a simple observation and go on to present some criteria for the planarity of the graph on edge insertion based on the type of the inserted edge.
\begin{observation}\label{obs:outerFace}
For a $3$-connected planar graph $G$, two planar embeddings $\mathcal{E}_1$ and $\mathcal{E}_2$ in the plane have the same vertex rotation scheme if between the two embeddings only the clockwise order of vertices on the boundary of the outer faces of the two embeddings are reverse if each other and all the clockwise order of vertices on the boundary of internal faces is same.
\end{observation}

Notice that due to the above fact, given a planar embedding with its outer face $F_0$ and an internal face $F_1$ specified, we can modify it to make $F_1$ the outer face while keeping the vertex rotation scheme unchanged by just reversing the orientation of the faces $F_0$ and $F_1$. 

Next, we present some criteria to determine if an edge to be inserted in a planar
graph causes it to become non-planar.

\begin{restatable}{lemma}{ThToThPlanCrit}\label{lem:ThToThPlanCrit}
	For any $3$-connected planar graph $G$, $G+\{\{a,b\}\}$ is planar if and only if $a$ and $b$ lie on the boundary of a common face. 
\end{restatable}

Next, let us consider the case where the vertices $a$ and $b$ are in the same block, say $B_i$, of a connected component of the graph. Let $R_a,R_b\in V(T_3(B_i))$ be two R-nodes in the SPQR-tree of $B_i$ such that $a\in V(R_a)$ and $b\in V(R_b)$. Consider the path between $R_a$ and $R_b$ nodes in $T_3(B_i)$, $R_a,P_1,R_1,P_2,\ldots,R_k,P_{k+1},R_b$, where $R_i$ and $P_i$ are R-nodes and S-nodes in $T_3(B_i)$ respectively, that appear on the path between $R_a,R_b$ (see \cref{fig:two3}). We have the following lemma (see Section~\ref{sec:app:graphThApp} for the proof). 
\begin{restatable}{lemma}{ptest}\label{lem:ptest}
	$G+\{\{a,b\}\}$ is planar if and only if
	\begin{alphaenumerate}
		\item all vertices in $\{a\}\cup V(P_1)$ lie on a common face boundary in the embedding of $R_a$
		\item all vertices in $V(P_{k+1})\cup \{b\}$ lie on a common face boundary in the embedding of $R_b$, and
		\item for each $i \in [k]$ all vertices in $V(P_i)\cup V(P_{i+1})$ lie on a common face in the embedding of $R_i$. Equivalently  $R_a\sqline R_b$ tree path is a coherent path.
	\end{alphaenumerate}
\end{restatable}

Consider the case in which $\es[+]{2}{3}$ edge $\{a,b\}$ can be inserted into $G$ preserving planarity. Notice that the $3$-connected components $R_a,R_1,\ldots,R_k,R_b$ coalesce into one $3$-connected component after the insertion of the edge. Let this coalesced $3$-connected component be $R_{ab}$. Obviously, $\{a,b\}$ would be at the boundary of exactly two faces of $R_{ab}$, say $F_0,F_1$. We claim that the separating pair vertices in $\bigcup_{i\in[k+1]} V(P_i)$ all lie either on the boundary of $F_0$ or $F_1$. See Figure~\ref{fig:twoCol}. We defer the proof of the following lemma to Section~\ref{sec:app:graphThApp}.

\begin{restatable}{lemma}{pCol}\label{lem:p2Col}
	The faces $F_0$ and $F_1$ define a partition into two parts on the set of vertices in the separating pairs $\bigcup_{i\in[k+1]} V(P_i)$  such that, for all $i\in[k+1]$ the two vertices in $V(P_i)$ belong to different blocks of the partition.
\end{restatable}

Finally, if the vertices $a$ and $b$ are in the same connected component but not in the same biconnected component then we use the following lemma to test for edge insertion validity. Let the vertices $a,b$ lie in a connected component $C_i$. Let $B_a, B_b\in N(T_2(C_i))$ be two block nodes in the BC-Tree of the connected component $C_i$ such that $a\in V(B_a)$ and $b\in V(B_b)$. Consider the following path between $B_a$ and $B_b$ in $T_2(C_i)$, $B_a,c_1,B_1,c_2,\ldots,c_k,B_k,c_{k+1},B_b$ where $B_i$ and $c_i$ are block and cut nodes respectively, in $T_2(C_i)$ that appear on the path between $B_a$ and $B_b$ (see \cref{fig:ot}). We abuse the names of cut nodes to also denote the cut vertex's name. Insertion of such an edge, i.e, $\es[+]{1}{2}$ leads to the blocks $B_a,B_1,\ldots,B_k,B_b$ coalescing into one block, call it $B_{ab}$. In the triconnected decomposition of $B_{ab}$ a new cycle component is introduced that consists of the edge $\{a,b\}$ and virtual edges between the consecutive cut vertices $c_i,c_{i+1}$, $i\in[k]$. See Figure~\ref{fig:otsp}. We defer the proof of the following lemma to Section~\ref{sec:app:graphThApp}.
\begin{restatable}{lemma}{planBTest}\label{lem:planBTest}
	$G+\{\{a,b\}\}$ is planar if and only if
		(a) $G[V(B_a)]+\{\{a,c_1\}\}$ is planar,
		(b) $G[V(B_b)]+\{\{c_{k+1},b\}\}$ is planar, and
		(c) for each $i\in [k]$, $G[V(B_i)]+\{\{c_i,c_{i+1}\}\}$ is planar. 
\end{restatable}

\section{Query Model and Auxiliary Relations}\label{sec:query}
In this section we formalize the query model and the auxiliary data stored. 

\paragraph*{Query Model}
We describe here what is meant by \emph{maintaining} a planar embedding.
Given a planar embedding of $G$, for every vertex in $G$, a clockwise cyclic order of its neighbour vertices as per the embedding is called the \emph{vertex rotation scheme}, and for every face in the planar embedding of $G$, the clockwise cyclic order of vertices on the face is called the \emph{face-vertex rotation scheme}. Maintaining a \emph{planar embedding} of the graph implies that we can answer the first query below, while maintaining an \emph{extended} planar embedding implies we can answer both the queries below:
\begin{romanenumerate}
	\item For a vertex $v_i\in V(G)$ do its neighbour vertices, $v_{i_1},v_{i_2},v_{i_3}\in V(G)$ appear in the clockwise cyclic ordering $v_{i_1},v_{i_2},v_{i_3}$ in the vertex rotation scheme?
	\item  Do three vertices $v_{i_1},v_{i_2},v_{i_3}\in V(G)$ lie on a common face in the clockwise cyclic order $v_{i_1},v_{i_2},v_{i_3}$ in the face-vertex rotation scheme?  
\end{romanenumerate}

\paragraph*{Auxiliary Relations}\label{subsec:auxRel}
We begin with a necessary definition. 
\begin{definition}[Coherent Path]\label{def:coherent}
A SPQR-tree path between two nodes $X_1,X_k$ is said to be a coherent path if
each $X_i$ on the path,  is either an R-node or  an S-node and
for every rigid component $R_i$ along the path, the two P-nodes
$P_i,P_{i+1}$ incident on $R_i$  satisfy that there exists an
embedding of the graph induced by $R_i$ such that the vertices
constituting the separating pairs $P_i,P_{i+1}$ lie on a 
common face of the embedding. 
\end{definition}

The following are the auxiliary relations that we maintain:
\begin{romanenumerate}
	\item Betweenness on BC-trees: this relation tells us if a particular node of a BC-tree lies in between two given nodes on a path on the tree.
	\item Betweenness on SPQR-trees: same as above for SPQR-trees.
	\item Extended planar embedding of the S,R-nodes of each SPQR-tree. This consists of maintaining the vertex rotation scheme and face-vertex rotation scheme restricted to vertices and faces of the triconnected components.
	\item For every coherent path of an SPQR-tree there is a 
              	two-colouring of the vertices of all the P-nodes along the path
		such that the two vertices of a P-node get different colours. See Figure~\ref{fig:twoCol}.  
	\item A vertex rotation scheme for one  planar embedding (of potentially exponentially many) for each B-node of all the BC-trees.
\end{romanenumerate}
\begin{figure}[h]
	\centering
	\includegraphics[width=0.5\textwidth]{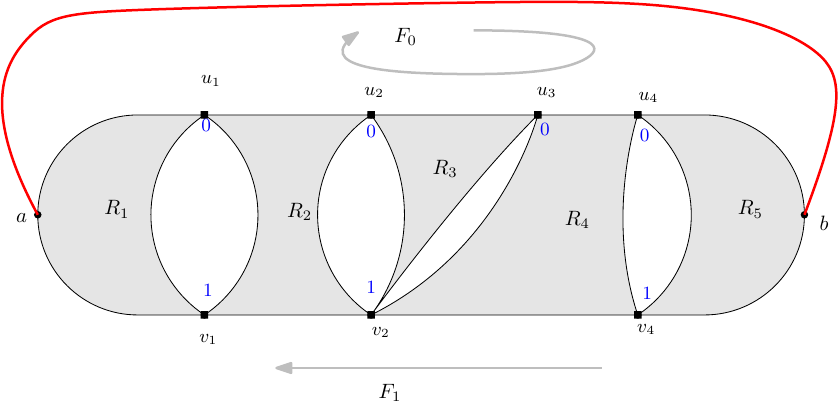}
	\caption{Two-colouring of P-node vertices.}
	\label{fig:twoCol}
\end{figure}

\section{Dynamic Biconnectivity and Triconnectivity}
\label{sec:app:triconn}
We know that connectivity can be maintained in $\DynFO$ due to \cite{PI}. That is we can query in $\FO$ at every step whether any two vertices are in the same connected component. We can enrich the connectivity relation to maintain connectivity in the graphs $G - \{u\}, G - \{u,v\}$ for all $u,v\in V(G)$ as follows. We can augment $\FO$ predicates used in the dynamic connectivity algorithm of \cite{PI} with two extra arguments $u,v$ to indicate that graph of interest is $G-\{u, v\}$ and in that any access to the input graph edge relation in the formula can be modified to ignore vertices $u, v$ and edges incident on them.        
   
Having access to connectivity relation in all these graphs we can maintain the BC-tree and SPQR-tree relations in $\FO$ by almost following their graph-theoretic definitions verbatim. For example, to test if a vertex $v$ is a cut vertex, we check if there are two vertices that are connected in $G$ and not connected in $G-\{v\}$. We further explain in detail other related relations to BC-tree and SPQR-tree. We also give brief justifications for their definability in $\FO$. 

\subsubsection*{BC-Tree related $\FO$ Primitives}
\subparagraph*{Cut vertex} A vertex $w$ is a cut vertex if there exist two vertices $u,v$ such that $u$ and $v$ are connected in $G$ but not in $G-\{w\}$. This can be decided by accessing the vertex connectivity relations in $G$ and $G-\{w\}$. Since we know that connectivity is in $\DynFO$ it follows that we can maintain in $\FO$ whether a vertex is a cut vertex or not.

\subparagraph*{Block} Two vertices are in the same block in the biconnected decomposition of the graph iff they remain connected despite removal of any one vertex from the graph. So $u$ and $v$ are in the same block iff $u$ and 
$v$ are connected $G-\{w\}$ for all $w\notin \{u,v\}$.

\subparagraph*{Block name} We know that no two vertices can belong 
to more than one block, i.e, there is a unique block that 
two biconnected vertices lie in. As a consequence, we can 
identify each block by the lexicographically least ordered 
pair of vertices $(a,b)$ that lie in it. So, a pair of 
vertices $(a,b)$ is a block name iff (1) $a$ and $b$ lie 
in a common block and (2) amongst all pairs of vertices 
$(u,v)$ such that $u$ and $v$ are biconnected and lie in 
the same block as $(a,b)$, $(a,b)$ is the 
lexicographically least one. Clearly both (1) and (2) are 
$\FO$ definable given that \emph{Block} is $\FO$ definable.

\subparagraph*{Betweenness on BC-tree} 
Given two blocks $B_1,B_2$ by their unique identifiers and a cut vertex $u$, we can decide whether $u$ lies in between the $B_1\sqline B_2$ BC-tree path as follows. If there exist two vertices $v$ and $w$ such that $v\in V(B_1)$ and $w\in V(B_2)$ and $v,w$ are connected in $G$ but not in $G-\{u\}$. Similarly, we can decide betweenness for any triple $(*_1,*_2,*_3)$, i.e, whether $*_2$ lies in between $*_1\sqline*_3$ BC-tree path (each of $*_1,*_2,*_3$ could be either cut vertex or a block name).   

From the above discussion we obtain the following:
\begin{lemma}\label{lem:BC}
	The BC-tree of each connected component of an undirected graph along with the betweenness relation can be maintained in $\DynFO$ under edge updates.
\end{lemma}

\subsubsection*{SPQR-Tree related $\FO$ Primitives}
Before we describe SPQR-tree related primitives, we make a note here that we only consider $3$-connected separating pairs for the purpose of computing the triconnected decomposition of a graph following~\cite{DLNTW,HR2}, unlike the the construction of Hopcroft and Tarjan~\cite{HT73}.

\subparagraph*{Separating Pair}
A pair vertices, $\{s,t\}$ form a $3$-connected separating pair iff (1) $s$ and $t$ are $3$-connected, i.e., there are $3$-vertex disjoint paths between $s$ and $t$ and (2) there exist two vertices $u,v\notin\{s,t\}$ such that $u$ and $v$ are connected in $G$, $G-\{s\}$ and $G-\{t\}$ but not in $G-\{s,t\}$. (1) can be checked in $\FO$ using connectivity queries and Menger's Theorem~\cite[Theorem 3.3.5]{Diestel} as follows. (1) is true iff $s$ and $t$ are connected in $G-\{u,v\}$ for all vertices $u,v\notin\{s,t\}$. (2) is simply checking connectivity in three graphs.           

\subparagraph*{Triconnected components and names} 
Recall that both rigid and cycle components are called triconnected components. A triple of vertices $a,b$ and $c$ belong to a common triconnected component iff $a,b$ and $c$ are connected in $G-\{s,t\}$ for all separating pairs $\{s,t\}$. Moreover, to further decide whether the common triconnected component is a rigid component (R) we can just check that $a,b$ and $c$ are connected in $G-\{p,q\}$ for all pairs of vertices $\{p,q\}$ (not just the separating pairs). Otherwise, the common triconnected component is a cycle S component. Also, three vertices belong to at most one common triconnected component. As a consequence, we can uniquely identify each triconnected component by the lexicographically least ordered triple of vertices $(u,v,w)$ that lie in it. So, a triple of vertices $(u,v,w)$ is a triconnected component name if (1) $u,v$ and $w$ are in the same triconnected component and (2) amongst all $(x,y,z)$ that lie in the same triconnected component as $(u,v,w)$, $(u,v,w)$ is the lexicographically the least one. 

\subparagraph*{Betweenness on SPQR-tree} Given two triconnected nodes $W_1,W_2$ by their names and a separating pair $P_i$, we can decide whether $P_i$ lies in between the 
$W_1\sqline W_2$ SPQR-tree path as follows. If there exist two vertices $u,v$ such that $u\in V(W_1)$ and $v\in V(W_2)$ and $u$ and $v$ are connected in $G$ but not in $G-V(P_i)$. Similarly, betweenness for any three arbitrary nodes (S, P, R nodes) of the SPQR-tree can be decided.   

From the above discussion we obtain the following.
\begin{lemma}\label{lem:SPQR}
	The SPQR-tree of each biconnected component of an undirected graph
	along with the betweenness relation maintained in $\DynFO$ under edge updates.
\end{lemma}

\section{Dynamic Planar Embedding: Algorithm Overview}\label{sec:algoOverview}
Our idea is to maintain planar embeddings of all triconnected components (S-nodes and R-nodes) of the graph and use those to find the embedding of the entire graph. Insertions and deletions of edges change the triconnected components of the graph, i.e., a triconnected component might decompose into multiple triconnected components or multiple triconnected components may coalesce together to form a single one. The same is true of biconnected components, i.e., a biconnected component might decompose into multiple biconnected components or multiple triconnected components may coalesce together to form a single biconnected component. 

We discuss here how we update the embeddings of the triconnected components under insertions and deletions, assuming that we have the SPQR-tree and BC-tree relations available at every step (which we have indicated how to maintain in Section~\ref{sec:app:triconn}).

Some of the edge insertions/deletions are easier to describe, for example if the edge is being inserted in a rigid component then only the embedding of that rigid component has to change to reflect the presence of the new edge and introduction of two new faces. Thus, we first establish some notation to differentiate between classes of edges for ease of exposition.
\begin{definition}
	A  graph is \emph{\maxcon~$i$-connected} if it is $i$-connected but is not $i+1$-connected for $i\in\{0,1,2\}$. For $i=3$, a graph is \maxcon~$i$-connected if the graph is $3$-connected. 
\end{definition}
\begin{definition}\label{def:edgeType}
	The \emph{type} of an edge is $\es[\sigma]{i}{j}$ where $i,j\in\{0,1,2,3\}$ and $\sigma\in\{+,-\}$ such that
	\begin{itemize}
		\item $\sigma=+$ if the edge is being inserted into $G$. $\sigma=-$ if the edge is being deleted.
		\item both the endpoints are in a common \maxcon~$i$-connected component before the change and in a common \maxcon~$j$-connected component after the change.
	\end{itemize}
\end{definition} 

In Table~\ref{tab:algo}, we summarize which type of edge update affects each of the auxiliary relations. For a description of the auxiliary relations see Section~\ref{subsec:auxRel} and for a definition of the edge types see Definition~\ref{def:edgeType}. The table is intended to serve as a map to navigate the different parts of the algorithm.
\begin{table}[h]
\centering
\resizebox{\columnwidth}{!}{%
\begin{tabular}{|c|c|c|c|c|c|c|c|c|}
\cline{1-9}
\multicolumn{1}{|c|}{AuxData} & \multicolumn{8}{|c|}{Type of Edge (Definition~\ref{def:edgeType})}\\
\cline{2-9}
(Section~\ref{subsec:auxRel}) & $\es[+]{0}{1}$ & $\es[+]{1}{2}$ & $\es[+]{2}{3}$ & $\es[+]{3}{3}$ & $\es[-]{1}{0}$ & $\es[-]{2}{1}$ & $\es[-]{3}{2}$ & $\es[-]{3}{3}$ \\
\cline{1-9}
BC-Tree & \ref{subsec:pOut}(\ref{it:pzo}) &\ref{subsec:pOut}(\ref{it:pot}),\ref{subsec:pot} & & &\ref{subsec:mOut}(\ref{it:moz}) &\ref{subsec:mOut}(\ref{it:mto}),\ref{subsec:mto} & & \\
\cline{1-9}
SPQR-Tree &\ref{subsec:pOut}(\ref{it:pzo})  &\ref{subsec:pOut}(\ref{it:pot}),\ref{subsec:pot} & \ref{subsec:pOut}(\ref{it:ptth}),\ref{subsec:ptth} & & \ref{subsec:mOut}(\ref{it:moz})  &\ref{subsec:mOut}(\ref{it:mto}),\ref{subsec:mto} & \ref{subsec:mOut}(\ref{it:mtht}),\ref{subsec:mtht} & \\
\cline{1-9}
S,R-Embed & &\ref{subsec:pOut}(\ref{it:pot}),\ref{subsec:pot} & \ref{subsec:pOut}(\ref{it:ptth}), \ref{subsec:ptth} & \ref{subsec:pOut}(\ref{it:pthth}), \ref{subsec:pthth}& & \ref{subsec:mOut}(\ref{it:mto}),\ref{subsec:mto}& \ref{subsec:mOut}(\ref{it:mtht}),\ref{subsec:mtht}&\ref{subsec:mOut}(\ref{it:mthth}),\ref{subsec:mthth} \\
\cline{1-9}
Two-Coloring & &\ref{subsec:pOut}(\ref{it:pot}),\ref{subsec:cpot} & \ref{subsec:pOut}(\ref{it:ptth}),\ref{subsec:cptth} & \ref{subsec:pOut}(\ref{it:pthth}),\ref{subsec:cpthth} & &\ref{subsec:mOut}(\ref{it:mto}),\ref{subsec:cmto} & \ref{subsec:mOut}(\ref{it:mtht}),\ref{subsec:cmtht} & \ref{subsec:mOut}(\ref{it:mthth}),\ref{subsec:cmthth} \\
\cline{1-9}
B-Embed & & \ref{subsec:embB},\ref{subsec:blockEmbed} & \ref{subsec:embB},\ref{subsec:blockEmbed} & \ref{subsec:embB},\ref{subsec:blockEmbed}& & \ref{subsec:embB},\ref{subsec:blockEmbed} & \ref{subsec:embB},\ref{subsec:blockEmbed} & \ref{subsec:embB},\ref{subsec:blockEmbed} \\
\cline{1-9}
\end{tabular}
}
\caption{AuxData affected by type. The (hyper-linked) entries refer to the algorithm overview (Section~\ref{sec:algoOverview}) and the detailed description (Sections~\ref{sec:triChange},~\ref{sec:coloring}, \ref{sec:embed}).}
\label{tab:algo}
\end{table}

In the next two subsections we outline the updates in the triconnected planar embedding relations as well as the two-colouring relations which are described in complete detail in Sections~\ref{sec:triChange} and~\ref{sec:coloring}.
\begin{figure}[ht]
	\begin{subfigure}{0.45\textwidth}
		\includegraphics[width=\textwidth]{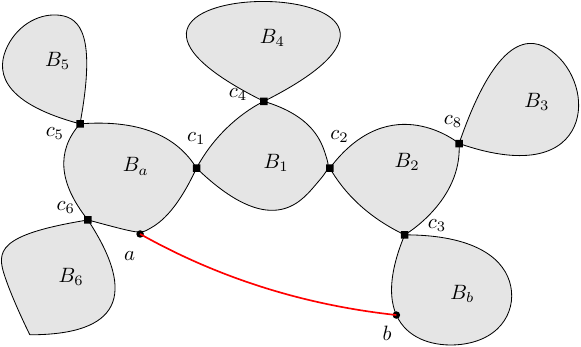}
		\caption{}
		\label{fig:otg}
	\end{subfigure}
	\begin{subfigure}{0.45\textwidth}
		\includegraphics[width=\textwidth]{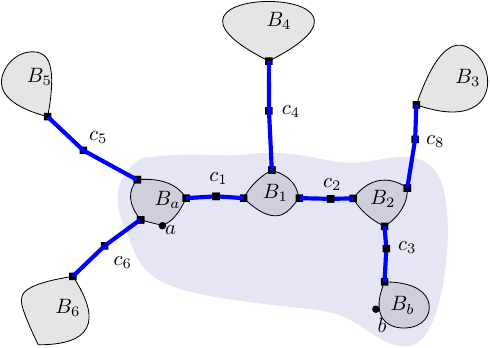}
		\caption{ }
		\label{fig:otbc}
	\end{subfigure}
	\begin{subfigure}{0.40\textwidth}
		\includegraphics[width=\textwidth]{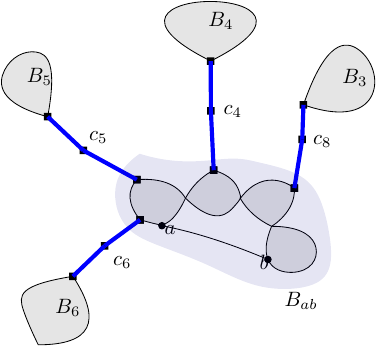}
		\caption{ }
		\label{fig:otbca}
	\end{subfigure}
	\hfill
	\begin{subfigure}{0.45\textwidth}
		\includegraphics[width=\textwidth]{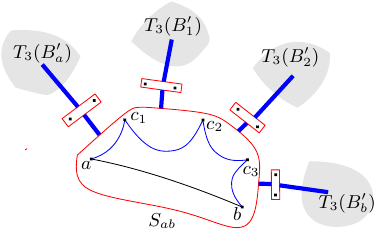}
		\caption{ }
		\label{fig:otsp}
	\end{subfigure}
	\caption{$\es[+]{1}{2}$ edge $\{a,b\}$ insertion. (\ref{fig:otg}) To-be-inserted edge is highlighted in red. (\ref{fig:otbc}) The path between block nodes $B_a$ and $B_b$ is highlighted in the BC-tree. (\ref{fig:otbca}) The BC-tree of the connected component $CC_i$. $B_a\sqline B_b$ path has shrunk to a single node $B_{ab}$ in the BC-tree. (\ref{fig:otsp}) The SPQR-tree of the block $B_{ab}$ contains a cycle component $S_{ab}$ made of the inserted edge and the new virtual edges that are introduced after its insertion.}
	\label{fig:ot}
\end{figure} 

\subsection{Edge insertion}\label{subsec:pOut}
Find the type $\es[+]{i}{j}$ of the inserted edge $\{a,b\}$ (where 
$i,j \in \{0,1,2,3\}$). This can be done using 
Lemmas~\ref{lem:BC},~\ref{lem:SPQR}. 
Depending on the type we branch to one of the following options:

\begin{alphaenumerate}
	\item $\es[+]{0}{1}$\label{it:pzo} This case affects only the BC-tree and SPQR-tree relations. Embedding and colouring relations remain unaffected.
	\item $\es[+]{1}{2}$\label{it:pot}: 
	In this case, both the endpoints $a,b$ are in the same 
	connected component $C$ but not in the same 
	biconnected component $B_a, B_b$ (see 
	Figure~\ref{fig:otg}). On this insertion, the BC-tree 
	of the connected component changes. All the 
	biconnected components on the $B_a\sqline B_b$ path 
	in the BC-tree of $C$ coalesce into one biconnected 
	component $B_{ab}$ (see \cref{fig:otbca}). In 
	the BC-tree of the connected component, for each pair of consecutive cut-vertices on the $B_a\sqline B_b$ path, a virtual edge is 
	inserted in the biconnected component shared by the 
	cut-vertex pair. This edge is a $\es[+]{2}{3}$ or 
	$\es[+]{3}{3}$ edge and is handled below. Since the 
	biconnected components involved are distinct, all 
	these edges can be simultaneously inserted.
	In addition, a cycle component is introduced for which 
	the face-vertex rotation scheme is computed via the 
	betweenness relation in the BC-tree (see \cref{fig:otsp}). The two-colouring relation is updated by extending the colouring of old biconnected components across the new cycle component (in \cref{fig:otsp}, a coherent path from a P-node in $T_3(B'_a)$ to a P-node in $T_3(B'_2)$ will have to go via an S-node $S_{ab}$). 

	\item $\es[+]{2}{3}$\label{it:ptth}: Both $a,b$ are in the same biconnected component, say $B_i$, but not in the same rigid component. Let $a\in V(R_a)$ and $b\in R_b$, where $R_a$ and $R_b$ are two R-nodes in the SPQR-tree of $B_i$. The SPQR-tree of the biconnected component changes after the insertion as follows. All the rigid components on the $R_a\sqline R_b$ SPQR-tree path coalesce into one rigid component (see Figures~\ref{subfig:two3b},~\ref{subfig:two3spqr}). The embedding of the coalesced rigid component is obtained by combining the embeddings of the old triconnected components that are on the $R_a\sqline R_b$ path, with their correct orientation computed from the two-colouring of the separating pairs for the corresponding coherent path. To update the two-colouring of the separating pair vertices, first we discard those old coherent paths and their two-colouring, that are no longer coherent as result of the insertion of $\{a,b\}$. While for the subpaths of the old coherent paths that remain coherent we obtain their two-colouring from that of the old path by ignoring colourings of old P-nodes on the $R_a\sqline R_b$ path.  
  
	\item $\es[+]{3}{3}$\label{it:pthth}: In this case, both the vertices are in the same $3$-connected component. Due to this insertion, connectivity relations do not change. We identify the unique common face in the embedding of the $3$-connected component that the two vertices lie on. We split the face into two new faces with the new edge being their common edge, i.e, the face vertex rotation scheme of the old face is split across the new edge. In the vertex rotation scheme of the two vertices, we insert the new edge in an appropriate order.  
\end{alphaenumerate}

\subsection{Edge deletion}\label{subsec:mOut}
Find the type of the edge $\{a,b\}$ that is being deleted. Depending on the type we branch to one of the following options:
\begin{alphaenumerate}
	\item $\es[-]{1}{0}$\label{it:moz}: This case occurs when both $a,b$ are in the same connected component but not in the same biconnected component. The edge is 
a cut edge and only BC-tree and SPQR-tree relations are affected.
	\item $\es[-]{2}{1}$\label{it:mto}: In this case, both vertices are in the same biconnected component but not in the same $3$-connected component before the change. This is the inverse of the insertion operation of an edge of the type $\es[+]{1}{2}$. The biconnectivity of the old biconnected component changes and it unfurls into a path consisting of multiple biconnected components that are connected via new cut vertices (in \cref{fig:otbca} if the edge $\{a,b\}$ is deleted, the block $B_{ab}$ will unfurl into the highlighted path in \cref{fig:otbc}). Old virtual edges that are still present in the new biconnected components are deleted. These virtual edge deletions are of type $\es[-]{3}{3}$ or $\es[-]{3}{2}$ which we describe how to handle in the following cases.

\item $\es[-]{3}{2}$\label{it:mtht}: In this case $a,b$ are in the same rigid component $R_{ab}$ (see Figure~\ref{subfig:two3a}) before the change. But after the deletion of $\{a,b\}$, $R_{ab}$ decomposes into smaller triconnected components that are unfurled into a path in the SPQR-tree (see \cref{subfig:two3spqr}).  
We compute an embedding of the triconnected fragments of $R_{ab}$ using the embedding of $R_{ab}$. From the updated embedding of $R_{ab}$, we compute the two-colouring of the separating pair vertices pertaining to the unfurled path on an SPQR-tree. We compose the obtained colouring with the rest of $B_{ab}$. The two-colouring of the elongated coherent paths is made consistent by flipping the colouring of the P-nodes of one subpath, if necessary. 
\item $\es[-]{3}{3}$\label{it:mthth}: In this case both $a,b$ are in the same $3$-connected component before and after the deletion of $\{a,b\}$. Connectivity relations (connectedness, bi/tri-connectivity) remain unchanged. The two faces in the embedding of the $3$-connected component that are adjacent to each other via the to-be-deleted edge are merged together to form a new face. From the vertex rotation schemes of the two vertices, the edge is omitted. Pairs of coherent paths that merge due to the merging of two faces have to be consistently two-coloured.
\end{alphaenumerate} 
Notice that we do not describe changes of types $\es[+]{0}{2},\es[+]{0}{3}, \es[+]{1}{3},\es[-]{3}{0},\es[-]{3}{1},$ $ \es[-]{2}{0}$, since edge insertion or deletion only changes the number of disjoint paths between any two vertices at most by one and thus these type of edge change are not possible. Also, we omit the changes of type $\es[+]{2}{2}$ and $\es[-]{2}{2}$ (related to S-nodes) as they are corner cases that we discuss in Section~\ref{subsec:cptt}.

\begin{figure}[ht]
	\centering
	\begin{subfigure}{0.42\textwidth}
		\centering
		\includegraphics[width=\textwidth]{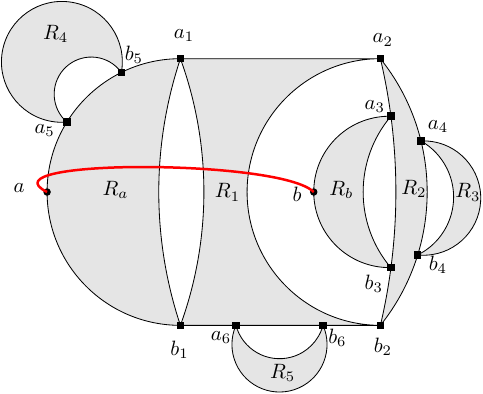}
		\caption{ }
		\label{subfig:two3b}
	\end{subfigure}
	\hspace*{2mm}
	\begin{subfigure}{0.53\textwidth}
		\centering
		\includegraphics[width=\textwidth]{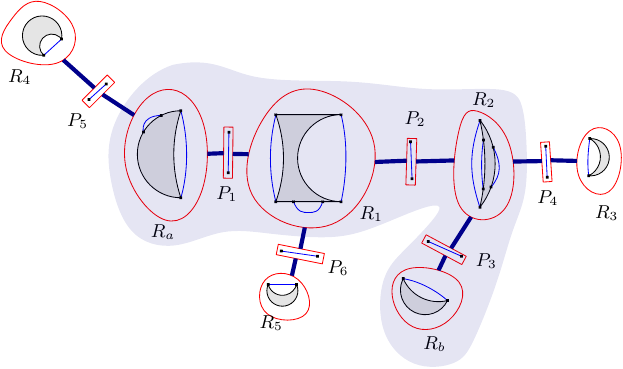}
		\caption{ }
		\label{subfig:two3spqr}
	\end{subfigure}
	\hspace*{2mm}
	\begin{subfigure}{0.5\textwidth}
		\centering
		\includegraphics[width=\textwidth]{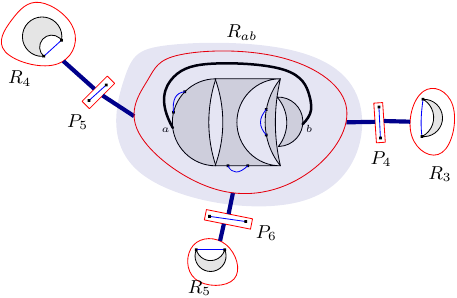}
		\caption{ }
		\label{subfig:two3a}
	\end{subfigure}
	\caption{A $\es[+]{2}{3}$ edge $\{a,b\}$ insertion. (\ref{subfig:two3b}) To-be-inserted edge is highlighted in red. (\ref{subfig:two3spqr}) The SPQR-tree of block $B_i$ before insertion. The path between R-nodes $R_a$ and $R_b$ is highlighted. (\ref{subfig:two3a}) SPQR-tree of $B_i$ after insertion. The $R_a\sqline R_b$ path has shrunk to a single node $R_{ab}$.}
	\label{fig:two3}
\end{figure} 

\subsection{Embedding the biconnected components}\label{subsec:embB}
For each biconnected component, we maintain an embedding in the form of the vertex rotation scheme for its vertices. Note that we do not need to maintain the face-vertex rotation scheme for the biconnected components. The relevant updates for maintaining the biconnected components embedding are edge changes of type $\es[+]{1}{2},\es[+]{2}{3},\es[+]{3}{3},\es[-]{3}{3},\es[-]{3}{2}$, and $\es[-]{2}{1}$. For a $\es[+]{2}{3}$ change, the part of the block that becomes triconnected, after the insertion of the edge, we patch together its vertex rotation scheme with unchanged fragments of the existing vertex rotation scheme in the affected block. For example, in~\cref{fig:two3}, after the insertion of edge $\{a,b\}$, R-nodes $R_a, R_1, R_2$ and $R_b$ coalesce together to form $R_{ab}$ (part of the block that has become triconnected). The embedding of the unchanged fragments of the block, viz., $R_3, R_4$ and $R_5$ is patched together at the appropriate P-nodes with the updated embedding of $R_{ab}$ to update the embedding of the entire block.  In a $\es[+]{1}{2}$ change, multiple biconnected components coalesce together. First, we update the vertex rotation scheme of each block by inserting a required virtual edge and then patch together the vertex rotation schemes at the old cut vertices. For a $\es[-]{3}{2}$ change, we splice in new virtual edges, that arise as a result of the change, in appropriate places in the vertex rotation scheme of the affected block. Details of the updates required in all types of edge changes are described in~\cref{subsec:blockEmbed}.

\subsection{Embedding the entire graph}\label{subsec:embG}
Assuming that we have planar embedding (in terms of the vertex rotation scheme) of each biconnected component of the graph we can compute a planar embedding of the whole graph as follows. For the vertices that are not cut vertices, their vertex rotation scheme is the same as their vertex rotation scheme in the embedding of their respective blocks that they belong to. For a cut vertex $v$, we join together its vertex rotation schemes from each block $v$ belongs to. We only need to take care that in the combined vertex rotation scheme of $v$, its neighbours in each block appear together and are not interspersed.

We now complete the proof of the main theorem modulo exposition that we defer to the Sections~\ref{sec:triChange},~\ref{sec:coloring} and \ref{sec:embed}.
\begin{restatable}{lemma}{triChange}\label{lem:triChange}
The extended planar embeddings of all cycle, rigid components and a planar embedding for each biconnected component of the graph can be updated in 
$\DynFO$ under edge changes of type: $\es[+]{1}{2}, \es[+]{2}{3}, \es[+]{3}{3},\es[+]{2}{2},\es[-]{2}{2}, \es[-]{2}{1}, \es[-]{3}{2}, \es[-]{3}{3}$.  
\end{restatable}    
The above lemma combined with Section~\ref{sec:embed} allows us to prove the main theorem:
\main*
\begin{proof}
If an edge insertion is within a $3$-connected component the criterion in
Lemma~\ref{lem:ThToThPlanCrit} tells us when the resulting component becomes
non-planar.
Analogously, Lemma~\ref{lem:ptest} informs us when a $2$-connected component
is non-planar on an edge insertion within it (but across different 
$3$-connected components). Similarly Lemma~\ref{lem:planBTest} allows us to
determine if an edge added within a connected component but between two
different $2$-connected components causes non-planarity. 
An invocation of Lemma~\ref{lem:triChange} shows how ot update the planar
embedding for biconnected graphs. Moving on to connected graphs, vertex rotation schemes of the cut-vertices is obtained by combining their vertex rotation schemes in each block in any non-interspersed order, using $\FO$ primitive \emph{merge vertex rotation schemes} (see \cref{fo:merge_vrs}). Remaining vertices inherit their vertex rotation scheme from the unique block they belong to. This completes the proof.
\end{proof}

\section{Details on the Graph Theoretic Machinery}\label{sec:app:graphThApp}
In this section we restate and prove the lemmas from Section~\ref{sec:graphTh}.

\ThToThPlanCrit*
\begin{proof}
	Assume that $G+\{\{a,b\}\}$ is planar and the $a$ and $b$ are not on the boundary of a common face in the embedding of $G$. Since $G+\{\{a,b\}\}$ is also $3$-connected and planar graph, its planar embedding is unique up to reflection. The edge $\{a,b\}$ is adjacent to exactly two faces in $G+\{\{a,b\}\}$. If we remove the curve corresponding to $\{a,b\}$ from this embedding, we get an embedding of $G$ such that $a$ and $b$ are on the boundary of the common face. This implies that there are at least two embeddings for $G$ one in which $a$ and $b$ lie on a common face boundary and one in which they do not. But since $G$ is $3$-connected, this contradicts that $G$ has a unique planar embedding. 
\end{proof}  

\ptest*
\begin{proof}
	Assume that $G+\{\{a,b\}\}$ is planar. Consider the graph $G[V(R_i])$ along with the virtual edges corresponding to the separating pairs $P_i$ and $P_{i+1}$. Let $V(P_i) = \{u_1,v_1\}$, $V(P_{i+1} = \{u_2,v_2\})$. 
	\begin{claim}
		In $G+\{\{a,b\}\}$ there exists a tree $T$ rooted at $a$, such that $V(T)\cap V(R_i) = \{u_1,v_1,u_2,v_2\}$. and $E(T)\cap E(R_i)=\emptyset$.
Thus, each vertex in $V(T) \cap V(R_i)$ is a leaf of $T$.
	\end{claim}
	\begin{claimproof}
		Consider the subgraph of $G$, called $H_1$, on the following subset of vertices, $V(R_a)\cup \bigcup_{j\in[i-1]}V(R_j)$. $V(H_1)\cap V(R_i)=\{u_1,v_1\}$ since $\{V(P_i) = \{u_1,v_1\}\}$ is separating pair. Since $H_1$ is biconnected (by definition), there are two vertex disjoint paths from $a$ to $u_1$ and $v_1$ (due to Menger's Theorem). 
		Call the union of the two paths as a tree $T_a$ rooted at $a$, with $u_1$ and $v_1$ as its leaves. Since $T_a$ is a subgraph of $H_1$ and contains $u_1,v_1$, we have that $V(T_a)\cap V(R_i)=\{u_1,v_1\}$.
		Consider the subgraph of $G$, called $H_2$, on the following subset of vertices, $V(R_b)\cup \bigcup_{j\in[i+1,k]}V(R_j)$. $V(H_2)\cap V(R_i) = \{u_2,v_2\}$ since $V(P_{i+1})$ is a separating pair. Since $H_1$ is biconnected (by definition), there are two vertex disjoint paths from $b$ to $u_2$ and $v_2$. Call the union of the two paths as a tree $T_b$ rooted at $b$, with $u_2$ and $v_2$ as its leaves. Since $T_b$ is a subgraph of $H_2$ and contains $u_2,v_2$, we have that $V(T_b)\cap V(R_i)=\{u_2,v_2\}$. In $G+\{\{a,b\}\}$, the two trees $T_a$ and $T_b$ can be joined together at $a$ via the edge $\{a,b\}$, to give the required tree $T$. $V(T)\cap V(R_i)=\{u_1,v_1,u_2,v_2\}$ because $V(T_a)\cap V(R_i)=\{u_1,v_1\}$ and $V(T_b)\cap V(R_i)=\{u_2,v_2\}$. And $E(T)\cap E(R_i)=\emptyset$ since $\{u_1,v_1,u_2,v_2\}$ are leaves in $T$. 
\end{claimproof}
	Now, consider the subgraph of $G+\{\{a,b\}\}$, $R'_i = R_i\cup T$. $R'_i$ must be planar Since $G+\{\{a,b\}\}$ is planar. Consider a planar 
embedding of $R'_i$. If we remove the curves corresponding to the edges of $T$ 
we get an embedding of $R_i$. With respect to this embedding of $R_i$ the vertex $a$ must lie inside some face of this embedding of $R_i$. Let this face be $F_a$. We claim that all the vertices in $\{u_1,v_1,u_2,v_2\}$ must lie on the boundary of the face $F_a$. Suppose for the sake of 
contradiction there is a vertex, say $w \in \{u_1,v_1,u_2,v_2\}$, that does not lie on the boundary of $F_a$. Since $w \in V(R_i)$ it cannot lie inside a face of $R_i$, hence it lies outside of the boundary of $F_a$. There is a path $P$ in $T$ from $a$ to $w$ -- one of which is inside $F_a$ and 
the other, outside. Thus, the path must pass through a vertex, say  $x \in V(F_a)$. This is by the Jordan Curve Theorem \cite[Theorem 4.1.1]{Diestel} and since if $P$ crosses an edge 
of $F_a$ without passing through a vertex of $F_a$ would imply that the embedding $R'_i = R_i \cup T$ is non planar. Thus $x \in V(R_i) \cap V(T)$ is a 
vertex of degree at least two in $T$ because it is an internal vertex of a path $P$ in $T$. But by assumption $V(R_i)\cap V(T) = \{u_1,v_1,u_2,v_2\}$ and each of these four vertices is a leaf of $T$ so $x \not \in \{u_1,v_1,u_2,v_2\}$ which is a contradiction.
	 
	Similarly it can be shown that the vertices in $\{a\}\cup V(P_1)$ must lie on the boundary of a common face in the embedding of $R_a$. Also, that the vertices in $\{b\}\cup V(P_{k+1})$ must lie on the boundary of a common face in the embedding of $R_b$.
	Now we will show the other direction, i.e, if the three conditions in the statement of the lemma are met then $G+\{\{a,b\}\}$ is planar. Suppose that all the conditions (a),(b) and (c) in the lemma statement are met. Consider the subgraph of $G$ that consists of the triconnected components $R_a,R_1,\ldots,R_k,R_b$, call it $G_{ab}$. Then we claim the following.  
	\begin{claim}
		There exists a planar embedding of $G_{ab}$ such that the separating pair vertices $\bigcup_{i\in[k+1]}V(P_i)$ and the vertices $a$ and $b$ lie on the boundary of the outer face. 
	\end{claim}
	\begin{proof}
		We prove this by induction on the path $R_a, R_1, \ldots, R_k, R_b$. 
		For any $i\in[k]$ define the graph $G_{a,i}$ to be the subgraph of $G$ consisting of the triconnected components $R_a,R_1,\ldots,R_i$. 
		Induction Hypothesis: There exists an embedding of $G_{a, i}$ such that the separating pair vertices $\bigcup_{j\in[i+1]}V(P_j)$ and the vertex $a$ lie on the boundary of the outer face.
		
		Base case: $i=1$. We know that there is an embedding of $R_a$ such that $a$ and $V(P_1)$ lie on its outer face boundary and there is an embedding of $R_1$ such that the vertices $b$ and $V(P_1)\cup V(P_2)$ lie on its outer face boundary. Moreover, in the embedding of $R_a$, the virtual edge corresponding to the separating pair $P_1$ lies on the face boundary. Also, in the embedding of $R_1$ the virtual edges corresponding to $P_1$ and $P_2$ lie on the boundary of the outer face. We can join the two embeddings at the common edge corresponding to $P_1$ after appropriately deforming them so that $a$ and $V(P_1)\cup V(P_2)$ are together on the boundary of the outer face of the combined embedding.
		
		Induction step: $l=i+1$. We know that there is an embedding of $R_{i+1}$ such that the separating pair vertices $V(P_{i+1}\cup V(P_{i+2}))$ lie on the boundary of its outer face and there is an edge corresponding to the separating pair $P_{i+1}$ on the outer face boundary. We can join the embeddings of $R_{i+1}$ with the drawing of $G_{a, i}$ at the common edge corresponding to $P_{i+1}$ after appropriately deforming them. The vertices $a$ and $\bigcup_{j\in[i+2]}V(P_j)$ lie on the boundary of the outer face of the combined embedding. 
		
		So by induction, we have proved the fact that $G_{a,k}$ has an embedding such that all the separating pair vertices in $P_1,\ldots, P_{k+1}$ along with the vertex $a$ lie on the boundary of the outer face. A similar argument will show that $G_{ab}$ has a planar embedding such that vertices in $P_1,\ldots, P_{k+1}$ along with $a$ and $b$ lie on the boundary of the outer face. 
	\end{proof}
	Since $G_{ab}$ has an embedding in which $a$ and $b$ are on the outer face, we can draw the curve for the edge $\{a,b\}$ in the outer face region so that it doesn't cross any edge or vertex of $G_{ab}$ except for the vertices $a$ and $b$. This implies that $G_{ab}+\{\{a,b\}\}$ is planar the embedding of $G_{ab}$ along with curve for $\{a,b\}$ is its valid planar embedding. To prove that $G+\{\{a,b\}\}$ is planar we invoke the fact that a graph is planar iff all its triconnected components are planar (\cref{lem:macl}). This condition is satisfied in $G+\{\{a,b\}\}$ since after the insertion of $\{a,b\}$ the triconnected components that do not lie on the path between $R_a$ and $R_b$ remain unchanged and the triconnected components that are on the path all coalesce together to form $G_{ab}+\{\{a,b\}\}$ which we have proved is planar.	
\end{proof}

Consider the case in which $\es[+]{2}{3}$ edge $\{a,b\}$ can be inserted into $G$ preserving planarity. Notice that the $3$-connected components $R_a, R_1, \ldots, R_k, R_b$ coalesce into one $3$-connected component after the insertion of the edge. Let this coalesced $3$-connected component be $R_{ab}$. Obviously, $\{a,b\}$ would be at the boundary of exactly two faces of $C_{ab}$, say $F_0,F_1$. We claim that the separating pair vertices in $\cup_{i\in[k+1]} V(P_i)$ all lie either on the boundary of $F_0$ or $F_1$. See \cref{fig:twoCol}.

\pCol*
\begin{proof}
	We show that there exists a non-separating induced cycle in $G[V(R_{ab})]$ such that exactly one vertex from each separating pair $P_i$ is on the cycle. From Tutte's Theorem, we know that non-separating induced cycles in a $3$-connected planar graph are precisely the faces of the graph. Since $G[V(R_{ab})]$ is $3$-connected and planar, the lemma follows.
	\begin{claim}
		There exists a non-separating induced path from $a$ to $b$ in $G[V(R_{ab})]$ that contains exactly one vertex from each separating pair $P_i,i\in[k+1]$. 
	\end{claim}
	\begin{claimproof}
		Induction Hypothesis: Assume that for an $i\in[k]$ there is a non-separating induced path from $a$ to a vertex in $V(P_i)$.
		
		Base Case: $V(P_1)=\{u_1,v_1\}$. Due to Lemma~\ref{lem:ptest} $a, u_1,v_1$ lie on the boundary of a common face in the embedding of $R_a$. Suppose they appear in the cyclic order $a, u_1,v_1$ on the face. Since $R_1$ is $3$-connected, the cycle corresponding to the face is induced non-separating due to Tutte's Theorem and since the $a$ to $u_1$ path on the face is a subgraph of the cycle, $a$ to $u_1$ path is non-separating induced path in $G[V(R_a)]$. The path is induced even in $R_{ab}$, and non-separating because after the path is removed from $G[V(R_{ab})]$, $R_a$ remains connected to the rest of $G[V(R_{ab})]$ via $v_1$, the other vertex in the separating pair $P_1$.
		
		Induction: By the induction hypothesis, suppose there is an induced non-separating path from $a$ to a vertex in $V(P_i)$, say $u_i$. Vertices in $V(P_i)\cup V(P_{i+1})$ are on the boundary of a common face in the embedding of $R_i$ due to Lemma~\ref{lem:ptest}. 
		The vertices $u_i,u_{i+1},v_{i+1},v_i$ must appear either in that cyclic order on the face or in the order $u_i,v_{i+1},u_{i+1},v_i$ since $\{u_i,v_i\}$ and $\{u_{i+1},v_{i+1}\}$ are virtual edges in $R_i$. In the latter case, there is a non-separating induced path from $u_i$ to $v_{i+1}$ in $R_i$ since the path is a subgraph of the face cycle. The path remains induced and non-separating even in $R_{ab}$ since after we remove the path from $R_{ab}$, $R_i$ remains connected to the rest of $R_{ab}$ via the vertices $v_i$ and $u_{i+1}$. The path is induced because all the edges that could be there between two vertices of the path will also have to be in $R_i$ alone, which is not possible. In the other case, when they lie in the cyclic order $u_i,u_{i+1},v_{i+1},v_i$, similarly we can establish that there is a non-separating induced path from $u_i$ to $u_{i+1}$. Thus, in any case, the $a$ to $u_i$ non-separating induced path can be extended to either $u_{i+1}$ or $v_{i+1}$. Hence the induction hypothesis follows.  
	\end{claimproof}
	Due to the above claim, there is a non-separating induced path from $a$ to a vertex in $V(P_{k+1})$, say $u_{k+1}$ in $R_{ab}$. But, since the vertices in $V(P_{k+1})$ and $b$ lie on the boundary of a common face in the embedding of $R_b$, using the same argument as in the base case of the above induction, we conclude that there is an induced non-separating path from $u_{k+1}$ to $b$ in $R_b$. Again this path can be combined with the path from $a$ to $u_{i+1}$ along with the edge $\{a,b\}$ resulting in an induced non-separating cycle in $R_{ab}$. And this induced non-separating cycle contains exactly one vertex from each $P_i$. Due to Tutte's theorem, hence they all lie on the boundary of a common face in the embedding of $R_{ab}$ since it is $3$-connected and planar. Similarly, it can be shown that the vertices from the separating pairs that are not on this cycle are in another non-separating induced cycle.    
\end{proof}

Finally, if the vertices $a$ and $b$ are in the same connected component but not in the same biconnected component then we use the following lemma to test for edge insertion validity. Let the vertices $a,b$ lie in a connected component $C_i$. Let $B_a, B_b\in N(T_2(C_i))$ be two block nodes in the BC-Tree of the connected component $C_i$ such that $a\in V(B_a)$ and $b\in V(B_b)$. Consider the following path between $B_a$ and $B_b$ in $T_2(C_i)$, $B_a,c_1,B_1,c_2,\ldots,c_k,B_k,c_{k+1},B_b$ where $B_i$ and $c_i$ are block and cut nodes respectively, in $T_2(C_i)$ that appear on the path between $B_a$ and $B_b$. We would abuse the names of cut nodes to also denote the cut vertex's name.
Insertion of such an edge, i.e, $\es[+]{1}{2}$ leads to the blocks $B_a,B_1,\ldots,B_k,B_b$ coalescing into one block, call it $B_{ab}$. In the triconnected decomposition of $B_{ab}$ a new cycle component is introduced that consists of the edge $\{a,b\}$ and virtual edges between the consecutive cut vertices $c_i,c_{i+1}$, $i\in[k]$.
\planBTest*
\begin{proof}
	Assume that the graph $G+\{\{a,b\}\}$ is planar. 
	\begin{claim}
		The graph $G[V(B_i)]+\{\{c_i,c_{i+1}\}\}$ is planar for all $i\in[k]$.
	\end{claim}
	\begin{claimproof}
		Before the edge $\{a,b\}$ insertion there exists a path from $a$ to $c_i$ and a path from $b$ to $c_{i+1}$ such that these two paths are vertex disjoint. After the insertion, this implies that there is path $Q$ from $c_i$ to $c_{i+1}$ in $B_{ab}$ such that $V(Q)\cap V(B_i)=\{c_i,c_{i+1}\}$. Consider the subgraph of $B_{ab}$ that consists of the path $Q$ and $B_i$. Since $B_{ab}$ is planar, this subgraph must also be planar. But, if the $B_i\cup Q$ is planar, $B_i+\{\{c_i,c_{i+1}\}\}$ should also be planar since the $Q$ could be replaced with just the edge $\{c_i,c_{i+1}\}$ via edge contractions. 
	\end{claimproof} 
	Similarly, $B_a+\{\{a,c_1\}\}$ and $B_b+\{c_{k+1},b\}$ must be planar.
	Now we will prove the other direction, that is, if $B_i+\{\{c_i,c_{i+1}\}\}$ is planar for all $i\in[k]$ and $B_a+\{\{a,c_1\}\}$ and $B_b+\{\{b,c_{k+1}\}\}$ then $G+\{\{a,b\}\}$ is planar. Since $B_i$ is biconnected and $B_i+\{\{c_i,c_{i+1}\}\}$ is planar, there exists an embedding of $B_i$ such that $c_i$ and $c_{i+1}$ lie on the boundary of the outer face, as we have seen in the Lemma~\ref{lem:ptest}.   
\end{proof}

\section{Planar Embedding Related $\FO$ Primitives}\label{sec:planarFO}
In this section we describe some $\FO$ primitives that are used by our planar embedding algorithm. We also provide brief justification for their $\FO$ definability.
\subparagraph*{Face}\label{fo:face} 
In a triconnected component, any three of its 
vertices lie on at most one common face. As a consequence, 
we can identify each face of a triconnected component by 
the lexicographically least order triple of vertices 
$(u,v,w)$ that lie on the face boundary. For each face we 
store the face-vertex rotation scheme triples along with 
the corresponding face name.

\subparagraph*{Make outer face}\label{fo:oface}
For a triconnected component, 
given its extended planar embedding, and an specified 
inner face we can make it the outer face in $\FO$ as 
follows. Due to \cref{obs:outerFace}, we only 
need to flip all the face-vertex rotation scheme triples for the old outer face and the specified inner face. All other faces retain their old face-vertex rotation scheme. 
Suppose the outer face name is $F_1$ and the specified inner face name is $F_2$. In the new embedding, where the $F_2$ becomes the outer face, we include a triple of 
vertices say $(u,v,w)$ in the face-vertex rotation scheme corresponding to the face $F_2$ iff $(w,v,u)$ is a triple corresponding to face $F_2$ in the old embedding.          

\subparagraph*{Flip}\label{fo:flip} We can flip the current extended planar 
embedding of a triconnected component in $\FO$ similar to 
the case above. Just flip every triple in face-vertex 
rotation scheme as well as the vertex rotation scheme.  

\begin{figure}
	\centering
	\begin{subfigure}{0.23\textwidth}
		\centering
		\includegraphics[width=\textwidth]{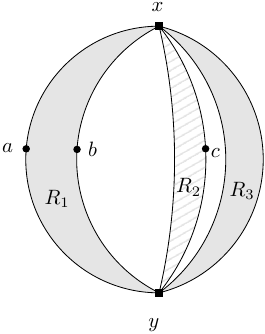}
		\caption{ }
		\label{fig:flipSlide}
	\end{subfigure}
	\hspace*{4mm}
	\begin{subfigure}{0.23\textwidth}
		\centering
		\includegraphics[width=\textwidth]{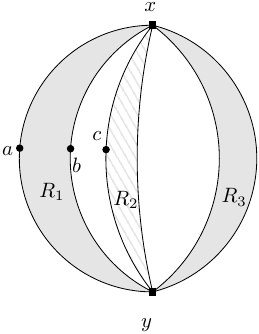}
		\caption{ }
		\label{fig:flips}
	\end{subfigure}
	\hspace*{4mm}
	\begin{subfigure}{0.30\textwidth}
		\centering
		\includegraphics[width=\textwidth]{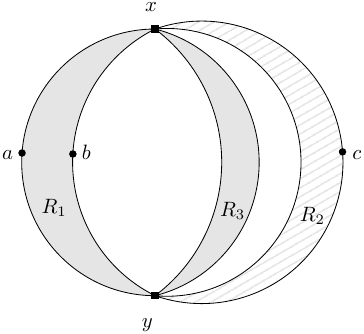}
		\caption{ }
		\label{fig:slides}
	\end{subfigure}
	\caption{Flips and slides in a biconnected component. (\ref{fig:flipSlide}) $R_1,R_2$ and $R_3$ are rigid components connected together via separating pair $\{x,y\}$. An embedding of the biconnected component. (\ref{fig:flips}) To bring vertices $b$ and $c$ on the same face, the embedding of $R_2$ is flipped.(\ref{fig:slides}) To bring the vertices $a$ and $c$ on the same face, we need to slide the embedding of $R_2$.}
	\label{fig:flipsSlides}
\end{figure}

\subparagraph*{Merge two faces}\label{fo:merge} Given two faces $F_1$ and $F_2$ 
(by their names) in a rigid component that are adjacent 
to each other via an edge $\{u,v\}$. If the edge $\{u,v\}$ 
is removed from the embedding then the faces $F_1$ and $F_2$
to form a new face say $F_{12}$. Then we can merge the 
face-vertex rotation scheme of $F_1$ and $F_2$ to get the 
face-vertex rotation scheme of $F_{12}$ in $\FO$ as follows.
First, we compute the name of the new face $F_{12}$ by finding
the smallest three vertices amongst the ones present in the
names of $F_1$ and $F_2$. This can be done $\FO$.
To compute the set of vertex triples that for the face-vertex 
rotation scheme, notice that all the triples of vertices that 
belonged to the face-vertex rotation scheme of $F_1$ or $F_2$ 
now belong to $F_{12}$. However there are new triples of vertices
to include, namely ones that have their vertices lying on both 
$F_1$ and $F_2$. Formally, a triple of vertices $(a,b,c)$ form a 
vertex triple in the face-vertex rotation scheme of $F_{12}$
iff either it is part of one $F_1,F_2$ or there exist two triples
$(w_1,w_2,z)$ and $(z,w_3,w_4)$ in the face-vertex rotation scheme 
of $F_1$ or $F_2$ such that $z\in\{u,v\}$ and 
in the cyclic order $w_1,w_2,z,w_3,w_4,w_1$, $(a,b,c)$ appears as 
a subsequence.     

\subparagraph*{Split a face}\label{fo:split} Given a face $F$, and two vertices $u,v$ that lie
on it, we can split the face into two faces by inserting the edge 
$\{u,v\}$ (we are assuming that $\{u,v\}$ was not an edge earlier). 
Contrary to the case merging operation here we need to filter out
some vertex triples from the face-vertex rotation scheme of $F$
to get the one for new faces, say $F_1,F_2$.
To compute the face-vertex rotation scheme of $F_1$ and $F_2$ 
we do the following. First, we compute the set of all vertices 
that lie on one face, say $F_1$. as follows. 
$S=\{w| (u,w,v) \text{ is a triple in the face-vertex rotation 
	scheme of } F\}$. We can define this set using a $\FO$ predicate
which take only one argument is true only for the vertices in $S$.
Then using the set $S$ we can filter the face-vertex rotation scheme 
triples for $F_1$. A vertex triple $(a,b,c)$ is in the 
face-vertex rotation of the new face $F_1$ iff it is a triple in the
face-vertex rotation scheme of $F$ and $\{a,b,c\}\subseteq S$.
Similarly, for $F_2$ we can compute the corresponding triples.
Also we can compute the names of $F_1$ and $F_2$ by computing 
the lexicographically smallest triple of vertices in $S\cup\{u,v\}$ 
and $(V(F)\setminus{S})\cup\{u,v\}$ respectively.

\subparagraph*{Merge vertex rotation schemes}\label{fo:merge_vrs}
We are given multiple cyclic orders on vertices, $C_1, \ldots, C_k$ in the form of triples. That is, each cyclic order $C_i$ is given as a set of ordered triple $\{(u,v,w)|u,v,w \text{ appear in that cyclic order in } C_i\}$ ($i\in [k]$). For each $C_i$ we are also given two consecutive vertices $t_i,s_i$ in the cyclic order. We have to merge the cyclic orders $C_1,\ldots, C_k$ into one cyclic order $C$ such that cyclic order is $s_1,\ldots, t_1,s_2,\ldots,t_2,\cdots, s_k,\ldots, t_k, s_1$. We can obtain the merged cyclic order $C$ in $\FO$ as follows. A triple of vertices $(x,y,z)$ is valid triple in $C$ if the vertices $x,y$ and $z$ belong to unique cyclic orders $C_{i_1}, C_{i_2}$ and $C_{i_3}$ respectively such that 
$i_1\le i_2\le i_3 \lor i_3\le i_1\le i_2 \lor i_3\le i_1\le i_2$. See Figure~\ref{fig:mergeCycles}.

\begin{figure}
	\centering
	\includegraphics[width=0.25\textwidth]{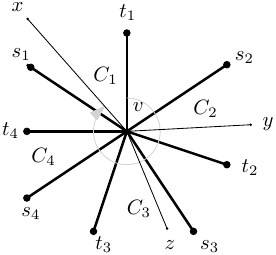}
	\caption{Merging vertex rotation schemes at $v$. Cyclic orders $C_1,C_2, C_3$ and $C_4$ are merged. $(x,y,z)$ is a valid triple in the merge cyclic order.}
	\label{fig:mergeCycles}
\end{figure}

If all the vertices in $(x,y,z)$ belong to a common cyclic order $C_i$ then it is a valid triple for $C$ iff it is a valid triple for $C_i$. If $x,y\in C_j$ and $z\in C_l$, ($j,l\in[k]$)  such that in $C_j$, $(s_j,x,y),(x,y,t_j)$  are valid triples then $(x,y,z)$ is a valid triple in $C$. If $x,z\in C_j$ and $y\in C_l$ such that in $C_j$, $(s_j,z,x),(z,x,t_j)$ are valid triples then $(x,y,z)$ is a valid triple in $C$. This covers all the cases. 

\section{Maintaining a Planar Embedding of Triconnected Components}
\label{sec:triChange}
\subsection{Handling an edge of type $\es[+]{3}{3}$}\label{subsec:pthth}
The edge $\{a,b\}$ that is to be inserted of type $\es[+]{3}{3}$, i.e, both the vertices $a$ and $b$ belong to a common $3$-connected component. First, we find out the id of the block that the vertices belong to and then the id of the rigid component, say $R_{ab}$, in that block that the vertices $a$ and $b$ lie in. To test the validity of $\{a,b\}$ insertion we look up the face-vertex rotation scheme of $R_{ab}$ and check if there exists a face such that the vertices $a$ and $b$ lie on its boundary, in $\FO$. If there is no such face then we reject the insert query due to the planarity criteria laid out in \cref{lem:ThToThPlanCrit}. Otherwise, there is such a face say $F_{ab}$ (wlog we can assume it to be an internal face). After the insertion of $\{a,b\}$, $F_{ab}$ splits into two new faces $F_a$ and $F_b$ that share the edge $\{a,b\}$. All the other faces of $R_{ab}$ remain unaffected. We use the $\FO$ primitive \emph{split a face} (see Section~\ref{sec:planarFO}) to compute the face-vertex rotation scheme pertaining to these faces to updating the face-vertex rotation scheme of $R_{ab}$.

We also need to update the vertex rotation scheme for $R_{ab}$. We observe that for the vertices in $R_{ab}$ other than $a$ and $b$, their vertex rotation scheme doesn't change after the insertion. We only need to update the vertex rotation scheme around $a$ and $b$. We splice in the vertex $a$ in the vertex rotation scheme for $b$ between the two neighbours of $b$ that lie on the boundary of the old face $F_{ab}$. We can find the two neighbours of $b$ that lie on the boundary of face $F_{ab}$ in $\FO$ by querying the face-vertex rotation scheme pertaining to $F_{ab}$. Similarly, we also update the vertex rotation scheme for $a$.

\subsection{Handling an edge of type $\es[+]{1}{2}$}\label{subsec:pot}
Let the vertices $a$ and $b$ belong to a connected component, say $C_1$. Wlog assume $a$ and $b$ are not cut vertices, so that they belong to unique blocks, say $B_a$ and $B_b$ respectively in BC-tree of $C_1$, $T_2(C_1)$.  Let the path from $B_a$ block node to $B_b$ block node in $T_2(C_1)$ be $B_a,c_1,B_1,...,B_k,{c}_{k+1},B_b$, where $B_i$ ($i\in [k]$) are block nodes and $c_i$ ($i\in[k+1]$) are cut vertex nodes. We can identify this path $\FO$ using the BC-tree related $\FO$ primitives described in~\cref{sec:app:triconn}. After the edge $\{a,b\}$ is inserted in the graph $G$, the blocks $B_a,B_b$ and $B_i$ ($i\in[k]$) coalesce together to form one biconnected component, say $B_{ab}$ (see \cref{fig:otbca}). Let the SPQR-tree of the block $B_{ab}$ be $T_3(B_{ab})$. 
\begin{proposition}
	$P_i = \{c_i,c_{i+1}\}$ for $i\in[k]$ are P-nodes in $T_3(B_{ab})$. Also, there is a cycle component node (S-node) in $T_3(B_{ab})$ that consists of the edge $\{a,b\}$ and the virtual edges corresponding to the separating pairs $\{c_i,c_{i+1}\}$.
\end{proposition}

The virtual edges $\{c_i,c_{i+1}\}$ become part of the biconnected component $B_{ab}$. All these virtual edges need to be inserted in $G[V(B_{ab})]$. Seemingly, we have to deal with insertion of arbitrarily many edges in one step. But, notice that the type of these edges is $\es[+]{2}{3}$ or $\es[+]{3}{3}$ and they are all in separate old blocks $B_a,B_b$ or $B_i$ ($i\in[k]$). Assuming that we can handle insertion of a single edge of the type $\es[+]{2}{3}$, we can handle insertion of all the above edges in one step in parallel. We defer the discussion of single edge insertion of the type $\es[+]{2}{3}$ to the next section.
 
We check the validity of the insertion of $\{a,b\}$ by using the planarity criteria in Lemma~\ref{lem:planBTest}.

\subsection{Handling an edge of type $\es[+]{2}{3}$}\label{subsec:ptth}
To handle an edge $\{a,b\}$ insertion such that $a$ and $b$ are in the same biconnected component, we do the following. Let the $B_i$ be the block in which both the vertices $a,b$ lie in, i.e, $a,b\in V(B_i)$. Let $T_3(B_i)$ be the triconnected decomposition tree of the block $B_i$. Let $R_a$ and $R_b$ be the R-nodes such that the vertices $a$ and $b$ lie in respectively, i.e, $a\in V(R_a)$, $b\in R_b$. Let $R_a, P_1, R_1, \ldots, R_k, P_{k+1}, R_b$ be the path between the $R_a$ and $R_b$ nodes on the triconnected decomposition tree such that each $R_i$, $i \in [k]$, is a R-node and each $P_i$, $(a_i,b_i)$, ($i\in[k+1]$) is a P-node. We can identify this path using the $\FO$ primitives related to SPQR-tree described in Section~\ref{sec:planarFO}. 
\begin{proposition}
	After the insertion of the edge $\{a,b\}$, all the vertices in $V(R_a)\cup V(R_b)\cup_{i\in[k]} V(R_i)$ are in a common rigid component. 
\end{proposition}
Let the new rigid component in the above lemma be $R_{ab}$. When the edge $\{a,b\}$ is inserted the rigid components $R_a, R_1, \ldots, R_k,R_b$ coalesce into the new rigid component, $R_{ab}$. See \cref{fig:two3}. We need to update the embedding of  this new rigid component, $R_{ab}$. It is done in the following manner:
\begin{enumerate}
	\item \label{it:al1} Identify the block $B_i$ in which the vertices $a$ and $b$ lie in. Identify the  R-nodes, $R_a$ and $R_b$ in the $T_3(B_i)$, the SPQR-tree of $B_i$, such that $a\in V(R_a)$ and $b\in R_b$.
	\item From the betweenness $\FO$ primitive (see Section~\ref{sec:app:triconn}) pertaining to $T_3(B_i)$, identify the triconnected component nodes and the P-nodes that are in-between $R_a$ and $R_b$ nodes. Let the $R_a\sqline R_b$ path on $T_3(B_i)$ be $R_a, P_1,R_1,\ldots R_k,P_{k+1},R_b$, where $R_i$ are the triconnected component nodes and the $P_i$ are the P-nodes. Check the validity of $\{a,b\}$ insertion by using the planarity criteria laid out in \cref{lem:ptest} and looking up the face-vertex rotation scheme for each of $R_i$ in $R_a,R_1,\ldots,R_k,R_b$. We only proceed if the test is passed, otherwise we discard the edge insertion query. 
	\item Finding the flips.
	\begin{alphaenumerate}
		\item We find the face in the embedding of $R_a$ that the vertices $a,a_1,b_1$ lie on its boundary. Let the face be $F_a$. Also, for $R_b$ we find the face (say $F_b$) that the vertices $b, a_{k+1},b_{k+1}$ lie on the boundary of. Similarly, for each $R_i$ ($i\in[k+1]$), we find the face $F_i$ that the vertices $V(P_i)\cup V(P_{i+1})$ lie on the boundary of.  Then we modify the extended planar embeddings of $R_a,R_b$ and $R_i$s($i\in[k+1]$) by making the  faces $F_a,F_b$ and $F_i$ their outer faces (if they are not already) respectively using the \emph{make outer face} $\FO$ primitive described in Section~\ref{sec:planarFO}. Now, Assume, without loss of generality, that the vertices $a,a_1,b_1$ lie in that clockwise order  on the boundary of the face $F_a$. From the two-colouring relation pertaining to the coherent path $P_1,P_2,\ldots,P_{k+1}$ find out the colour of the vertices $a_1$ and $b_1$. Assume that they are coloured $0$ and $1$ respectively.
		\item For each intermediate triconnected component $R_i$ choose the correct embedding to be composed together with the embedding of $R_a$ as follows. If the vertices $a_i,a_{i+1},b_i$ and $b_{i+1}$ are coloured as $0,0,1$ and $1$ respectively (see Lemma~\ref{lem:p2Col}) and the vertices appear in the clockwise order $a_i,a_{i+1},b_{i+1},b_i$ on the boundary of the face, then keep the old embedding. Otherwise, if they appear in the order $a_i,b_i,b_{i+1},a_{i+1}$, flip the current embedding (choose the reflected embedding) using the $\FO$ primitive \emph{flip} described in Section~\ref{sec:planarFO}. 
		\item  For $R_b$, choose its correct embedding as follows. The vertices $a_{k+1},b_{k+1},b$ lie on the boundary of face $F_b$. If the vertices $a_{k+1},b_{k+1}$ are coloured $0$ and $1$ respectively and the appear in the clockwise order $a_{k+1},b,b_{k+1}$ on the boundary of the face then keep the old embedding. Otherwise, if they appear in the clockwise order $a_{k+1},b_{k+1},b$ flip the embedding of $R_b$ using the $\FO$ primitive \emph{flip}.
	\end{alphaenumerate}
	 
	\item \textbf{New faces in the embedding of $R_{ab}$.} There are two special faces in the embedding of $R_{ab}$. say $F_0$ and $F_1$, such that all the separating pair vertices in $\cup_{i\in [k+1]}V(P_i)$ that are coloured $0$ appear on the boundary of $F_0$ and the vertices that are coloured $1$ appear on the boundary of $F_1$. 	
	Compute the face vertex rotation scheme for these two faces as follows. First, compute the face-vertex rotation scheme for the vertices on the boundary of the face $F_0$ restricted to separating pair vertices (coloured $0$) from the betweenness relation pertaining to $T_3(B_i)$ as follows: to decide whether separating pair vertices $u,v,w$ that are coloured $0$ appear in that clockwise order on the boundary of face $F_0$ check if the separating pairs $P_{i_1},P_{i_2},P_{i_3}$ (or any right cyclic shift of this sequence, e.g, $P_{i_2},P_{i_3},P_{i_1}$) appear in that order on the path between $R_a$ and $R_b$ in $T_3(B_i)$ such that $u\in V(P_{i_1}),v\in V(P_{i+2})$ and $w\in V(P_{i+3})$. If not, then conclude that $u,v,w$ do not appear in that clockwise order. Similarly, compute the face vertex rotation scheme for the face $F_1$. See \cref{fig:twoCol}. From the above computed face vertex rotation scheme restricted to separating pair vertices, compute the complete face vertex rotation scheme for the faces $F_0$ and $F_1$.
	\item For each pair of consecutive triconnected components $R_i,R_{i+1}$ there is new face in the embedding of $R_{ab}$. Compute the face vertex rotation scheme pertaining to such faces by merging the two faces using the $\FO$ primitive \emph{merge two faces} described in Section~\ref{sec:planarFO}. The faces to be merged are one of $R_i$ that has the vertices $a_i,b_i$ on its boundary but not both the vertices $a_{i-1}$ and $b_{i-1}$, and one of $R_{i+1}$ that has vertices $a_i,b_i$ on its boundary but not both the vertices $a_{i+1}$ and $b_{i+1}$.
	\item Compose the vertex rotation scheme at the separating pair vertices $a_i,b_i$  from the embeddings of all the rigid components that are incident on them using the $\FO$ primitive \emph{merge vertex rotation schemes} as follows. If $a_i$ is coloured $1$ then the vertex rotation scheme of $a_i$ in rigid components (on the $R_a\sqline R_b$ SPQR-tree path) that contain $a_i$ are merged in the order of closeness to node $R_a$ on the path. The closeness can be determined by accessing the betweenness related $\FO$ primitives. On the other, hand if $a_i$ is coloured $0$, then the vertex rotation schemes are merged in the order of closeness to the node $R_a$ on the path. Similarly, vertex rotation schemes at vertex $b_i$ are merged based on the colour. 
	\item Filter out the face vertex triples that are no longer valid. For example, triples that have vertices spanning faces $F_0$ and $F_1$.              
\end{enumerate}

We point out that the above description makes the following assumptions: (i) $a$ and $b$ are not part of any separating pair vertices in $T_3(B_i)$, i.e, rigid components that contain $a$ and $b$ are uniquely identified as $R_a$ and $R_b$ respectively (ii) there are no cycle components between $R_a$ and $R_b$ on the path in $T_3(B_i)$. To get rid of assumption (ii) we do the following. Consider a cycle component $S_j$ on the path between $R_a$ and $R_b$, such that the path is $R_a,\ldots, P_l, S_l, P'_l, \ldots, R_b$, where $P_l$ and $P'_l$ are the P-nodes that are present in $S_l$ as virtual edges. Let $V(P_l)=\{u,v\}$ and $V(P'_l)=\{w,x\}$. Let the cyclic order of these vertices on the cycle be $u,v,w,x$. Then modify the cycle component by adding the virtual edges $\{v,w\}$ and $\{x,u\}$ and splitting the cycle component at these virtual edges into two new cycle components. 
The vertices $u,v,w,x$ will belong to the R-node $R_{ab}$ after the insertion. Replace the $S_l$ node with the cycle on the vertices $u,v,S,w,x$ and treat it as a rigid component.

To get rid of assumption (i) we do the following. Suppose that the vertices $a$ and $b$ are part of a separating pair. They can't be part of the separating pair since the edge $\{a,b\}$ is a $\es[+]{2}{3}$ edge. Consider the set of all separating pairs of $B_i$ that contain $a$ as $X$ and the set of all separating pairs that contain $b$ as $Y$. Now, select the pair of separating pairs $P_x\in X, P_y \in Y$ such that there exists no $P'_x\in X$ and there exists no $P'_y\in Y$ such that neither $P'_x$ nor $P'_y$ appear on the path between $P_x$ and $P_y$ in $T_3(B_i)$. Having selected $P_x$ and $P_y$, identify the $R_a$ and $R_b$ as the triconnected component nodes such that they are neighbours of $P_x$ and $P_y$ respectively and they appear on the path between $P_x$ and $P_y$.  

\begin{figure}
	\centering
	\includegraphics[width=0.25\textwidth]{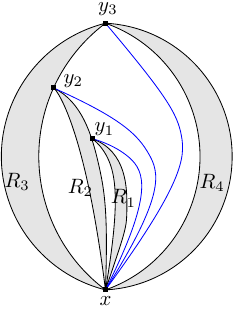}
	\caption{A biconnected component with rigid components $R_1,R_2,R_3$ and $R_4$. Virtual edges are in blue. If an edge is inserted between a vertex of $R_1$ and a vertex of $R_4$, all the virtual edges vanish. }
	\label{fig:deleteMulti}
\end{figure}

We now describe how to handle a deletion based on the type of the deleted edge $\{a,b\}$. 
\subsection{Handling an edge of type $\es[-]{3}{3}$}\label{subsec:mthth}
Recall that a $\es[-]{3}{3}$ edge is such that the endpoints of the edge are in a common $3$-connected component before as well as after the change. The updates required are straightforward. The SPQR-tree of the block to which the edge belongs to, doesn't change, by definition. Embedding of the R component, say $R_{ab}$, only changes in that the two faces are merged into one. We identify the two faces that the edge $\{a,b\}$ belongs to, say $F_a$ and $F_b$. We merge their face-vertex rotation scheme using the $\FO$ primitive \emph{merge two faces} described in Section~\ref{sec:planarFO}. 
If any of the faces $F_a$ and $F_b$ was the outer face of $R_{ab}$ in the previous step, we declare the new face $F_{ab}$ as the outer face of $R_{ab}$. The Vertex rotation scheme of vertices of $R_{ab}$ other than $a$ and $b$ doesn't change. To update the vertex rotation scheme of $a$ we simply drop the triples that contain the vertex $b$ from the vertex rotation scheme of $a$ before the change. Similarly, we update the vertex rotation scheme of $b$ by dropping triples that contain the vertex $a$.   
\subsection{Handling an edge of type $\es[-]{3}{2}$}\label{subsec:mtht}
Recall that a $\es[-]{3}{2}$ edge is such that its endpoints $a,b$ are in a common rigid component before the edge deletion but are in different rigid components and in the same block after the deletion. Let the common rigid component that $a$ and $b$ lie in be $R_{ab}$. Deletion of such an edge gives rise to new separating pairs and thus, new virtual edges. These virtual edges are between vertices $V(R_{ab})$. We can find these separating pairs from the updated SPQR-tree relations that we maintain, in $\FO$. Now, we merge the two faces of $R_{ab}$ that the edge $\{a,b\}$ lies on the boundary of to get a face $F_{ab}$ using the \emph{merge two faces} $\FO$ primitive. This merged face contains all vertices of new P-nodes on its boundary.    
After the deletion of $\{a,b\}$, $R_{ab}$ decomposes into multiple rigid components and we need to update their face-vertex rotation scheme. Those faces of any new rigid node that do not have any vertex of the new P-nodes incident on their boundary, their face vertex rotation scheme can be inherited from the embedding of $R_{ab}$ from the previous step. For the other faces in new R-nodes, we infer their face-vertex rotation scheme from the face-vertex rotation scheme of the faces of $R_{ab}$ by taking only those triples, that contain vertices of the R-node, into consideration.

The two-colouring of separating pair vertices relation needs to be updated, because after the deletion of $\{a,b\}$, new separating pairs are introduced. In the SPQR-tree of $B_i$ after the edge $\{a,b\}$ has been deleted, $R_{ab}$ unfurls into the path $R_a, P_1, R_1, \ldots, R_k, P_{k+1}, R_b$, where $R_a, R_1, \ldots, R_k, R_b$ are new triconnected components and $P_1, \ldots, P_{k+1}$ are new separating pairs, with the old P-nodes that were neighbours of $R_{ab}$, now are adjacent to appropriate new node amongst $R_a, R_1,\ldots, R_k, R_b$. See Figure~\ref{subfig:two3spqr}. Restricted to this unfurled path we update the two-colouring of these separating pair vertices using the face-vertex rotation scheme of $F_{ab}$. After that, any old path that passed through $R_{ab}$ node, it's two-colouring of P-node vertices needs consistently merging the colouring of the P-node vertices of the unfurled path.
 
\subsection{Handling an edge of type $\es[-]{2}{1}$}\label{subsec:mto}
Recall that a $\es[-]{2}{1}$ edge is such that its endpoints are in a common biconnected component before the deletion but in different biconnected components after the deletion and in the same connected component. Let the common biconnected component that the vertices $a$ and $b$ lie in be $B_{ab}$.  It is clear that such an edge must have been part of a cycle component.
\begin{proposition}
	The edge $\{a,b\}$ is part of a cycle component in the triconnected decomposition of $B_{ab}$.
\end{proposition} 
Deletion of the edge $\{a,b\}$ gives rise to new cut vertices in $G[V(B_{ab})]$. Let the BC-tree of $G[V(B_{ab})]-\{\{a,b\}\}$ be $T_2(G[V(B_{ab})]-\{\{a,b\}\})$. The tree is just a path $B_a, c_1, B_1, \ldots, B_k,$ $c_{k+1}, B_b$, where $B_a,B_1,\ldots, B_k,B_b$ are block nodes such that $a\in V(B_a)$ and $b\in V(B_b)$ and $c_1,\ldots,c_{k+1}$ are cut vertices. Then, we have the following.
\begin{proposition}
	There are virtual edges between $\{c_i,c_{i+1}\}$, $i\in [k]$ in the triconnected decomposition of $B_{ab}$ before the edge $\{a,b\}$ is deleted. Moreover, $c_i$ and $c_{i+1}$, for all $i\in[k]$ together lie in a $3$-connected component in the triconnected decomposition of $B_{ab}$.
\end{proposition} 

We need to delete these virtual edges after the deletion of $\{a,b\}$ since they no longer separating are pairs. Due to the above proposition, we have that these virtual edge deletions are of type $\es[-]{3}{2}$ or $\es[-]{3}{3}$, whose deletion we have seen how to handle in the previous subsections. One crucial thing to note is that although we might have to delete an arbitrary number of virtual edges in one step, all of them in different $3$-connected components and the updates required are local to each deletion, i.e., the updates due to one edge do not conflict with the updates due to another edge.

\subsection{Handling an edge of type $\es[-]{1}{0}$}\label{subsec:moz}
A $\es[-]{1}{0}$ edge has endpoints in a common connected component before deletion, but after the deletion they become disconnected. Such edge deletion is easy to handle since the edge is not part of any $3$-connected component and in itself is a block that is removed after the deletion.
\section{Two-colouring of P-node Vertices and the Updates Required}
\label{sec:coloring}
The two-colouring relation of the separating pair vertices of each new biconnected component after the change also needs to be updated. We now describe how to update this relation.

We consider the updates required to the two-colouring of separating pair vertices relation after insertion of edge $\{a,b\}$ in the graph as follows depending on the type of the edge.   

\subsection{Updating two-colouring under edge of type $\es[+]{3}{3}$}
\label{subsec:cpthth}
The change of type $\es[+]{3}{3}$ does not change the SPQR-tree of the block that the vertices $a$ and $b$ lie in, by definition, since no new P-nodes are introduced. Let the R-node in which the vertices $\{a,b\}$ lie in be $R_{ab}$. For a path on the SPQR-tree that doesn't go through the R-node $R_{ab}$, its two-colouring from the previous step remains valid even after the insertion. However, for the paths between two P-nodes that go through the $R_{ab}$ node, the two-colouring of the previous might become invalid for the following reason. If the two P-nodes on the path that are neighbours of the $R_{ab}$ node, say $P_1$ and $P_2$, their vertices $V(P_1) \cup V(P_2)$ no longer lie on the boundary of a common face in the embedding of $R_{ab}$. As a result of $\{a,b\}$ insertion, the path's two-colouring becomes invalid because any future edge insertion across this path would violate planarity due to Lemma~\ref{lem:ptest}. For such paths, we break it into two paths, one for each side of the path along the  $R_{ab}$ node. These two paths inherit the two-colouring from the old path, while the old path itself is expunged from the two-colouring relation. We can identify the paths for which this change is required by looking up the betweenness relation of the SPQR-tree to see if the $R_{ab}$ node lies on the path and by looking up the face-vertex rotation scheme of $R_{ab}$ to check if the vertices in two P-nodes adjacent to $R_{ab}$ on the path lie on the boundary of a common face, in $\FO$. 

\subsection{Updating two-colouring under edge of type $\es[+]{2}{3}$}
\label{subsec:cptth}
Let the block that the vertices $a$ and $b$ lie in be $B_i$. By definition, after the insertion of an $\es[+]{2}{3}$ type edge, the BC-tree of $G$ doesn't change. We know that the path on the $T_3(B_i)$, the SPQR-tree of $B_i$, changes such that the rigid components $R_a, R_1,\ldots R_k, R_b$ coalesce into one $3$-connected component node $R_{ab}$, i.e the path is replaced by one node $R_{ab}$ with the neighbourhood of all the nodes on this path consolidated at the node $R_{ab}$. For any path between two P-nodes in $T_3(B_i)$ after the insertion of $\{a,b\}$, that does not go through $R_{ab}$ node, the two-colouring for the path can be continued from the previous step. However, if the path does pass through $R_{ab}$, we need to check if the colouring of the path from the previous step remains valid or not, similar to the case of $\es[+]{3}{3}$ that we discuss earlier. Basically, we check whether the vertices in the two P-nodes, that are adjacent to the $R_{ab}$ node in $T_3(B_i)$, lie together on the boundary of a common face or not, in the embedding of $R_{ab}$. If they do, we persist with the two-colouring of the path from the previous step, otherwise, we break the path into two paths each of which inherits the two-colouring from the previous step. These tests can be done effectively in $\FO$ by querying the face-vertex rotation schemes. The paths on SPQR-tree are identifiable in $\FO$ due the $\FO$ primitives described in Section~\ref{sec:app:triconn}.

\subsection{Updating two-colouring under edge of type $\es[+]{1}{2}$}
\label{subsec:cpot}
Let the connected component that the vertices $a$ and $b$ lie in be $CC_i$. Let the BC Tree of the $CC_i$ be $T_2(CC_i)$. Let $B_a,B_b\in V(T_2(CC_i))$ such that $a\in V(B_a), b\in V(B_b)$. Let $B_a, c_1, B_1, \ldots, B_k, c_{k+1}, B_b$ be the path between $B_a$ and $B_b$ block nodes on the BC tree such that $B_i$ are block nodes and $c_i$ are cut vertex nodes. After the insertion of $\{a,b\}$ the blocks $B_a,B_1,\ldots,B_k,B_b$ coalesce into one block, $B_{ab}$.  We need to update the two-colouring relations for all the paths in the SPQR tree of $B_{ab}$, $T_3(B_{ab})$. We know that there is a S-node in $T_3(B_{ab})$, call it, $S_{ab}$ such that $\{a,b\}\in E(S_{ab})$ and $\{c_i,c_{i+1}\}\in E(S_{ab})$ for $i\in[k]$. For the paths between two P-nodes say $P_s$ and $P_t$ such that $V(P_s), V(P_t)\subset V(B_i)$ the two-colouring relations are updated after inserting the virtual edge $\{c_i,c_{i+1}\}$ in $B_i$ (an insertion of type $\es[+]{2}{3}$) using the procedure described earlier.

However, for paths between two P-nodes, $P_s$ and $P_t$, that lie in two different old blocks, say $B_p$ and $B_q$ ($p,q\in[k]$) respectively, we update the two-colouring as follows. Let the path be $P_s,\ldots,\{c_i,c_{i+1}\},S_{ab},\{c_j,c_{j+1}\},\ldots,P_t$. Notice that the S-node $S_{ab}$ is one the path. We know a two-colouring pertaining to the path $P_s,\ldots,\{c_i,c_{i+1}\}$ from the two-colouring relation in $B_p$. We also know a two-colouring of the path $\{c_j,c_{j+1}\}\ldots, P_t$ from the two-colouring relation in $B_q$. We combine these two-colourings consistently by looking at the vertex order of the cut vertices $c_i,c_{i+1},c_j, c_{j+1}$ on the cycle $S_{ab}$. We might need to flip the two-colouring of one path in this process. Since these tests only require looking up the betweenness relation of BC-trees and the embedding of S and R-nodes, we can effectively perform them in $\FO$. 

\subsection{Updating two-colouring under edge between two vertices of an S-node}
\label{subsec:cptt}
This is a special case in which the to-be-inserted edge is a chord in a cycle component. As a result of such an insertion, the cycle component in consideration breaks into two new cycles along the chord $\{a,b\}$ and the SPQR-tree of the block that the vertices $a$ and $b$ lie in changes by splitting of the S-node into a path of length three with two new S-node and the new separating pair $\{a,b\}$. For any path between two P-nodes in the SPQR-tree that does not pass through the split S-node, its two-colouring is continued from that in the previous step. For paths that do pass through the split S-node, we only need to consistently colour the vertices $a$ and $b$ because now $\{a,b\}$ appears as a new separating pair on the path. We can do this by looking up the old embedding of the S-node in $\FO$.  

\subsection{Updating two-colouring under edge of type $\es[-]{3}{3}$}
\label{subsec:cmthth}
Due to the deletion of the edge $\{a,b\}$ two separating pairs which were on different faces $F_a, F_b$ before the deletion can happen to be on the same face $F_{ab}$ afterwards. Then pairs of coherent paths can join together because of the merged face $F_{ab}$. To update the two-colouring of the vertices of the P-nodes along the newly formed coherent path we may need to flip the colourings of the P-node vertices of one of the constituent subpaths and this can be done in $\FO$.

\subsection{Updating two-colouring under edge of type $\es[-]{3}{2}$}
\label{subsec:cmtht}
Due to the deletion, the node $R_{ab}$ unfurls into a coherent path
in the SPQR-tree. Thus it can combine with other coherent paths incident on one of its nodes to yield new maximal coherent paths. The two-colouring of the P-node vertices along the new coherent paths can be gleaned from the old colouring as in the previous case. This can be easily accomplished in $\FO$.

\subsection{Updating two-colouring under edge of type $\es[-]{2}{1}$}
\label{subsec:cmto}
The block $B_{ab}$ unfurls into a new path in the BC-tree and the 
corresponding cycle component loses the edge $\{a,b\}$. This causes
coherent paths that passed through the corresponding S-node to
fragment into subpath(s). This makes it necessary to clean up the
two-colouring relation to accommodate for such changes. Also, due
to the unfurling of the block $B_{ab}$ into several blocks, the newly formed blocks lose the virtual edges between consecutive cut vertices along the unfurled path. This triggers further two-colouring changes of the preceding types in distinct blocks and is in $\FO$.

\section{Dynamic Planar Embedding: The Details}\label{sec:embed}
We first show how to extend the embedding of the triconnected components maintained in Section~\ref{sec:triChange}  to an embedding of the biconnected components -- notice in this last, by embedding we mean a vertex rotation scheme only. Subsequently, we show how to construct the embedding of the entire graph using these. 

\subsection{Maintaining combinatorial embedding of biconnected components}
\label{subsec:blockEmbed}
We show how to maintain an embedding of each block in terms of the vertex rotation scheme at each of its vertices. This includes keeping track of the virtual edges. Here, we discuss the updates required by the types of the changes (insertion/deletions). Let the edge being inserted/deleted be $\{a,b\}$
\paragraph*{Handling an edge of type $\es[+]{3}{3}$}
This case is straightforward because insertion of the edge doesn't require changing the embedding (flips or slides) of the block concerned, only that the vertex $a$ is to be spliced in at appropriate place in the vertex rotation scheme of $b$ and vice versa for $b$ in the vertex rotation scheme of $b$.
\paragraph*{Handling an edge of type $\es[-]{3}{3}$}
In this case also there is no change to the embedding (flips and slides) to the block that the edge $\{a,b\}$ belong to. We just drop the vertices $a$ and $b$ from the vertex rotation scheme of vertices $b$ and $a$ respectively.    
\paragraph*{Handling an edge of type $\es[+]{2}{3}$}
Let the block that the edge is being inserted to be $B_i$, such that $a\in V(R_a)$ and $b\in V(R_b)$, where $R_a$ and $R_b$ are R-nodes. Let the path between $R_a$ and $R_b$ nodes be $R_a,P_1,R_1,\ldots, R_k,P_{k+1},R_b$. We need to update the embedding of the block $B_i$. The BC-tree of $B_i$ doesn't change. However, the SPQR-tree of $B_i$ does change. The triconnected components on the path from $R_a$ node to $R_b$ node all coalesce into one triconnected component $R_{ab}$. The embedding of the block $B_i$ might have changed as a result of the edge insertion. 
If we remove the nodes from the SPQR-tree, corresponding to triconnected component nodes $R_a, R_1,\ldots, R_k, R_b$, the tree gets fragmented into a forest that has trees rooted at P-nodes and possibly isolated P-nodes. Let the trees be $T_1, T_2,\ldots,T_s$. We will use the structure of this forest to update the embedding of the biconnected component as follows. For each tree $T_j$ ($j \in [s]$) filter out from the embedding of $B_i$, the embedding restricted to the nodes of the tree $T_j$. For each $T_j$, in the embedding of the new node $R_{ab}$, splice in the vertex rotation scheme of the biconnected component restricted to $T_j$ at the virtual edge corresponding to the root P-node (the P-node through which $T_j$ was attached to the $R_a\sqline R_b$ path in the SPQR-tree of $B_i$ before insertion of $\{a,b\}$) of the tree $T_j$. 

To exemplify, suppose the virtual edge corresponding to the root separating pair of $T_j$ is $\{u,v\}$. In updating the embedding of $B_i$, vertex rotation scheme of any vertex of $T_j$ that is not in $\{u,v\}$ remains the same as in the previous step. For the vertices $u$ and $v$ we do the following. Let the clockwise order of neighbours of $u$ be $v,w_1,\ldots,w_r,v$ when restricted to $T_j$, call it $O_j$. Let the clockwise order of neighbours of $u$, when restricted to vertices of $R_{ab}$, be $x_1,\ldots,x_c,v,\ldots,x_1$, call it $O_{ab}$. Then after splicing in the clockwise order of the neighbours of $u$ restricted to $T_j$ into this, the clockwise order of neighbours of $u$ around it in $R_{ab}\cup T_j$ is $x_1,\ldots,x_c,w_1,\ldots, w_r,v,\ldots,x_1$. Basically, the clockwise order $O_j$ is spliced in the clockwise order $O_{ab}$ between the predecessor of $v$ and $v$. Having spliced in the clockwise order $O_j$ at $u$ this way we should splice in the clockwise order of neighbours of $v$ restricted to $T_j$ into the clockwise order of neighbours of $v$ restricted to $R_{ab}$ in the space between $u$ and its successor. The difference between the splicing of the orders at $u$ and $v$ is arbitrary. We can choose to splice in ``before the virtual edge'' at the smaller numbered vertex from $u$ and $v$ and ``after the virtual edge'' at the other vertex.
\paragraph*{Handling an edge of type $\es[-]{3}{2}$}
Suppose that $a$ and $b$ belong to the block $B_i$ and the rigid component $R_{ab}$. After the deletion of the edge $\{a,b\}$ new virtual edges are inserted in the block corresponding to the new separating pairs. In the vertex rotation scheme of the block $B_i$, we need to update these changes. But for the vertices that are not part of any new separating pair, their vertex rotation scheme does not change. We continue with their vertex rotation scheme in $B_i$ from the previous step. For the vertices $a$ and $b$, we drop the edge from the vertex rotation scheme around $a$ and $b$. Only for vertices that are part of newly created separating pairs we need to update the vertex rotation scheme. We update the vertex rotation scheme of such vertices in the following way. Suppose that $\{u, v\}$ is a new separating pair. We have the embedding of $R_{ab}$ from the previous step. With respect to this embedding the vertices $u$ and $v$ lie on the boundary of a unique face. Suppose that the neighbours of $u$ that lie on the face boundary are $u_1,u_2$ and they appear in that clockwise order around the vertex $u$ in the embedding of $R_{ab}$. Then we splice in $v$ in the vertex rotation scheme of $u$ in $B_i$ immediately before the vertex $u_2$. For $v$ the vertex rotation scheme in $B_i$ is updated similarly.
\paragraph*{Handling an edge of type $\es[+]{1}{2}$}
Under the insertion of an edge $\{a,b\}$ of type $\es[+]{1}{2}$ such that $a$ and $b$ belong to the connected component $CC_i$ and the blocks $B_a$ and $B_b$ respectively. After the insertion of $\{a,b\}$ the BC-tree of $CC_i$ changes such that multiple blocks coalesce together into one block $B_{ab}$. The vertex rotation scheme of vertices in the coalesced block needs to be updated. For the vertices which are not part of new virtual edges the vertex rotation scheme doesn't change and can be continued from the previous step. For the vertices that are part of new virtual edges, their vertex rotation scheme will change as follows. Consider the path between $B_a$ and $B_b$ nodes in the BC-tree, $B_a,c_1,B_1,\ldots,B_k,c_{k+1},B_b$. After the insertion of the edge $\{a,b\}$ virtual edges are inserted in $B_a, B_b$, and $B_i, i\in[k]$ between the consecutive cut vertices. Each such virtual edge insertion is an insertion of type $\es[+]{2}{3}$ or $\es[+]{3}{3}$ in each block. We have discussed how to update the vertex rotation scheme of blocks under such changes. At each cut vertex $c_i$ we splice together the vertex rotation scheme of $B_i$ and $B_{i+1}$ such that virtual edge $\{c_{i-1},c_i\}$ comes immediately before the virtual edge $\{c_i,c_{i+1}\}$ in the combined vertex rotation scheme at $c_i$. At the vertex $a$ the edge $\{a,b\}$ is spliced in immediately before the virtual edge $\{a, c_1\}$ in the vertex rotation scheme at $a$ in $B_a$, and at the vertex $b$ the edge $\{a,b\}$ is spliced in immediately after the virtual edge $\{c_{k+1},b\}$ in the vertex rotation scheme of $b$ in $B_b$.

\paragraph*{Handling an edge of type $\es[-]{2}{1}$}
Under the deletion of an edge of type $\es[-]{2}{1}$ the updates required to the embeddings of the blocks are similarly handled. Let the block that the edge $\{a,b\}$ lies in be $B_{ab}$. After the deletion of the edge $\{a,b\}$ the block decomposes into multiple new blocks.  We update embedding of the new blocks by first projecting down the previous step embedding of $B_{ab}$ to the new block in consideration and then removing the virtual edges that are no longer valid from each block. These virtual edge deletions are of type $\es[-]{3}{2}$ or $\es[-]{3}{3}$ which we have seen how to handle.   

\subsection{Proof of Lemma~\ref{lem:triChange}}
\triChange*
\begin{proof}
We start out by determining the type of the inserted edge using 
Lemmas~\ref{lem:BC},~\ref{lem:SPQR} because the first lemma allows us to determine whether two vertices belong to the same block; while the second one allows us to do the same for triconnected components. Next, we use the tests from Lemmas~\ref{lem:ThToThPlanCrit},~\ref{lem:ptest},~\ref{lem:planBTest} to determine whether the inserted edge causes the graph to remain planar depending on whether the type of the edge is $\es[+]{3}{3}, \es[+]{2}{3}$ or $\es[+]{1}{2}$. Only if the graph remains planar we follow the procedure detailed in Section~\ref{sec:triChange} (complemented by the maintenance of the two-colouring in Section~\ref{sec:coloring}) to update the auxiliary relations, where $\FO$ computability is also argued.

On the other hand, deletion does not cause a graph to become non-planar and we can follow the procedures discussed in Sections~\ref{sec:triChange},~\ref{sec:coloring} and~\ref{subsec:blockEmbed} related to deletions to update the auxiliary relations, where $\FO$ computabilty is also argued.
\end{proof}

\subsection{Embedding Connected Components}
To give an embedding of the connected components of the graph we combine together the embeddings of its biconnected components as follows. Consider a vertex $v$ of a connected component $C$. If $v$ is not a cut vertex, then it must belong to a unique block of $C$, say $B_v$. In that case, all the neighbours of $v$ also lie in the block $B_v$. The vertex rotation scheme for $v$ in the embedding of $C$ is the same as the vertex rotation scheme of $v$ in the embedding of $B_v$.

If $v$ is a cut vertex, then it lies in multiple blocks of $C$, say $B_{v_1},\ldots, B_{v_l}$. We have the vertex rotation scheme of $v$ pertaining to each block $B_{v_i}$ ($i\in[l]$). To get the vertex rotation scheme for $v$ pertaining to all its neighbours in $C$, we splice together its vertex rotation schemes in $B_{v_1},\ldots, B_{v_l}$ in that clockwise order. This order is arbitrary. We say that the lexicographically smaller named block's vertex rotation scheme comes before any lexicographically larger named block's vertex rotation scheme. This can be done in $\FO$ by accessing the BC-tree relation and the biconnected component's embedding relation.

\section{Conclusion}\label{sec:concl}
We show that planarity testing and embedding is in $\DynFO$ where we are able to ensure that an edge is inserted if and only if it does not cause the graph to become non-planar. This is potentially an important step in the direction of solving problems like distance and matching where only an upper bound of $\DynFO[\oplus]$ (where $\FO[\oplus]$ is $\FO$ with parity quantifiers) was known in planar graphs. This is because we might be able to improve the known bound making use of planar duality which presupposes a planar embedding. It might also make problems like max flow, graph isomorphism, and counting perfect matchings which are all statically parallelisable when restricted to planar graphs, accessible to a $\DynFO$ bound.
\bibliography{biblio}
\end{document}